\documentclass[preprint,review,12pt]{elsarticle}
\biboptions{sort&compress} 
\usepackage[nomarkers,nofiglist,notablist]{endfloat}
\usepackage{graphicx, epsfig, color}
\usepackage{amsmath, amssymb}

\newtheorem{theorem}{Theorem}
\newtheorem{proof}{Proof}

\newcommand\absT[1]{\left|#1\right|^2}
\newcommand\absF[1]{\left|#1\right|^4}

\journal{Elsevier Signal Processing Journal}
\begin{document}
%========================================================================================================================
\begin{frontmatter}
\title{Constant Modulus Algorithms Using Hyperbolic Givens Rotation}

\author[Canada] {A. Ikhlef}							\ead{aikhlef@ece.ubc.ca}
\author[Algeria]{R. Iferroujene}					\ead{redha.ifer@gmail.com}
\author[France] {A. Boudjellal\corref{cor1}}		\ead{abdelouahab.boudjellal@etu.univ-orleans.fr.}
\author[France] {K. Abed-Meraim\corref{cor2}}		\ead{karim.abed-meraim@univ-orleans.fr}
\author[Algeria]{A. Belouchrani.}					\ead{adel.belouchrani@enp.edu.dz.}

\address[Canada]	{ECE Dep., Univ. of British Columbia, 2356 Main Mall, Vancouver, V6T 1Z4, Canada.}		
\address[Algeria]	{EE Dep., Ecole Nationale Polytechnique, BP 182 EL Harrach, 16200 Algiers, Algeria.}	
\address[France]	{Polytech'Orleans, PRISME Laboratory, 12 Rue de Blois, 45067 Orleans, France.}	
%\fntext[fn1]{}
%\fntext[fn2]{}
\cortext[cor1]{Corresponding author.}
\cortext[cor2]{Principal corresponding author.}

\begin{abstract}
We propose two new algorithms to minimize the constant modulus (CM) criterion in the context of blind source separation. The first algorithm, referred to as Givens CMA (G-CMA) uses unitary Givens rotations and proceeds in two stages: prewhitening step, which reduces the channel matrix to a unitary one followed by a separation step where the resulting unitary matrix is computed using Givens rotations by minimizing the CM criterion. However, for small sample sizes, the prewhitening does not make the channel matrix close enough to unitary and hence applying Givens rotations alone does not provide satisfactory performance. To remediate to this problem, we propose to use non-unitary Shear (Hyperbolic) rotations in conjunction with Givens rotations. This second algorithm referred to as Hyperbolic G-CMA (HG-CMA) is shown to outperform the G-CMA as well as the Analytical CMA (ACMA) in terms of separation quality. The last part of this paper is dedicated to an efficient adaptive implementation of the HG-CMA and to performance assessment through numerical experiments.
\end{abstract}

\end{frontmatter}

%========================================================================================================================
%========================================================================================================================
\section{Introduction}
\label{Sec:Intro}
During the last two decades, Blind Source Separation (BSS) has attracted an important interest. The main idea of BSS consists of finding the transmitted signals without using pilot sequences or a priori knowledge on the propagation channel. Using BSS in communication systems has the main advantage of eliminating training sequences, which can be expensive or impossible in some practical situations, leading to an increased spectral efficiency. Several BSS criteria have been proposed in the literature e.g. \cite{Haykin_Bk_00, Comon_Bk}. The CM criterion is probably the best known and most studied higher order statistics based criterion in blind equalization \cite{Yang_98, Abrar_10, Amine_04, Labed_13} and signal separation \cite{Adel_96, Papadias_00, Papadias_04, Veen_Chap_05} areas. It exploits the fact that certain communication signals have the constant modulus property, as for example phase modulated signals. The Constant Modulus Algorithm (CMA) was developed independently by \cite{Godard_80, Treichler_83} and was initially designed for PSK signals. The CMA principle consists of preventing the deviation of the squared modulus of the outputs at the receiver from a constant. The main advantages of CMA, among others, are its simplicity, robustness, and the fact that it can be applied even for non-constant modulus communication signals.

Many solutions to the minimization of the CM criterion have been proposed (see \cite{Veen_Chap_05} and references therein). The CM criterion was first minimized via adaptive Stochastic Gradient Algorithm (SGA) \cite{Treichler_83} and later on many variants have been devised. It is known, in adaptive filtering, that the convergence rate of the SGA is slow. To improve the latter, the authors in \cite{Chen_04} proposed an implementation of the CM criterion via the Recursive Least Squares (RLS) algorithm. The author in \cite{Agee_86} proposed to rewrite the CM criterion as a least squares problem, which is solved using an iterative algorithm named Least Squares CMA (LS-CMA). In \cite{Veen_ACMA_96}, the authors proposed an algebraic solution for the minimization of the CM criterion. The proposed algorithm is named Analytical CMA (ACMA) and consists of computing all the separators, at one time, through solving a generalized eigenvalue problem. The main advantage of ACMA is that, in the noise free case, it provides the exact solution, using only few samples (the number of samples must be greater than or equal to $M^{2}$, where $M$ is the number of transmitting antennas). Moreover, the performance study of ACMA showed that it converges asymptotically to the Wiener receiver \cite{Veen_01}. However, the main drawback of ACMA is its numerical complexity especially for a large number of transmitting antennas. An adaptive version of ACMA was also developed in \cite{Veen_Chap_05}. More generally, an abundant literature on the CM-like criteria and the different algorithms used to minimize them exists including references \cite{Veen_Chap_05, Abrar_10, Yuan_10, Lamare_10, Lamare_11}.

In this paper, we propose two algorithms to minimize the CM criterion. The first one, referred to as Givens CMA (G-CMA), performs prewhitening in order to make the channel matrix unitary then, it applies successive Givens rotations to find the resulting matrix through minimization of the CM criterion. For large number of samples, prewhitening is effective and the transformed channel matrix is very close to unitary, however, for small sample sizes, it is not, and hence results in significant performance loss. In order to compensate the effect of the ineffective prewhitening stage, we propose to use Shear rotations \cite{Fu_06, Iferr_09}. Shear rotations are non-unitary hyperbolic transformations which allow to reduce departure from normality.  We note that the authors in \cite{Fu_06, Iferr_09, Souloumiac_09, Iferr_10} used Givens and Shear rotations in the context of joint diagonalization of matrices. We thus propose a second algorithm, referred to as Hyperbolic G-CMA (HG-CMA), that uses unitary Givens rotations in conjunction with non-unitary Shear rotations. The optimal parameters of both complex Shear and Givens rotations are computed via minimization of the CM criterion. The proposed algorithms have a lower computational complexity as compared to the ACMA. Moreover, unlike the ACMA which requires a number of samples greater than the square of the number of transmitting antennas, G-CMA and HG-CMA do not impose such a condition. Finally, we propose an adaptive implementation of the HG-CMA using sliding window which has the advantages of fast convergence and good separation quality for a moderate computational cost comparable to that of the methods in \cite{Agee_86, Papadias_04, Veen_Chap_05}.

The remainder of the paper is organized as follows. Section \ref{Sec:Formulation} introduces the problem formulation and assumptions. In Sections \ref{Sec:GCMA} and \ref{Sec:HGCMA}, we introduce the G-CMA and HG-CMA, respectively. Section \ref{Sec:AHGCMA} is dedicated to the adaptive implementation of the HG-CMA. Some numerical results and discussion are provided in Section \ref{Sec:Results}, and conclusions are drawn in Section \ref{Sec:Conclusion}.
\section{Problem Formulation}
\label{Sec:Formulation}
Consider the following multiple-input multiple-output (MIMO) memoryless system model with $M$ transmit and $N$ receive antennas:
\begin{equation}\label{Eq01}
    \mathbf{y}(n)=\mathbf{x}(n)+\mathbf{b}(n)=\mathbf{A}\mathbf{s}(n)+\mathbf{b}(n)
\end{equation}
where $\mathbf{s}(n)=[s_{1}(n), s_{2}(n), \ldots, s_{M}(n)]^{T}$ is the $M\times 1$ source vector, $\mathbf{b}(n)=[b_{1}(n),b_{2}(n),\ldots,b_{N}(n)]^{T}$ is the $N\times 1$ additive noise vector, $\mathbf{A}$ represents the $N\times M$ MIMO channel matrix, and $\mathbf{y}(n)=[y_{1}(n), y_{2}(n), \ldots, y_{N}(n)]^{T}$ is the $N\times 1$ received vector.

In the sequel, we assume that the channel matrix $\mathbf{A}$ is full column rank (and hence $N\geq M$), the source signals are discrete valued (i.e., generated from a finite alphabet), zero-mean, independent and identically distributed (i.i.d.), mutually independent random processes, and the noise is additive white independent from the source signals. Note that these assumptions are quite mild and generally satisfied in communication applications.

Our main goal is to recover the source signals blindly, i.e., using only the received data. For this purpose, we need to compute an $M\times N$ separation (receiver) matrix $\mathbf{W}$ such that $\mathbf{W}\mathbf{y}(n)$ results in the source signals, i.e.
\begin{equation}\label{Eq02}
\mathbf{z}(n)=\mathbf{W}\mathbf{y}(n)=\mathbf{W}\mathbf{A}\mathbf{s}(n)+\bar{\mathbf{b}}(n)=\mathbf{G}\mathbf{s}(n)+\bar{\mathbf{b}}(n)
\end{equation}
where $\mathbf{z}(n)=[z_{1}(n),z_{2}(n), \ldots, z_{M}(n)]^{T}$ is the $M\times1$ vector of the estimated
source signals, $\mathbf{G}=\mathbf{W}\mathbf{A}$ is the $M\times M$ global system matrix and $\bar{\mathbf{b}}(n)=\mathbf{W}\mathbf{b}(n)$ is the filtered noise at the receiver output. Ideally, in BSS, matrix $\mathbf{W}$ separates the source signals except for a possible permutation and up to scalar factors\footnote{To remove these ambiguities, when necessary, side information or a short training sequence is always required.}, i.e.
\begin{equation}\label{Eq03}
    \mathbf{W} \mathbf{x}(n)=\mathbf{P} \mathbf{\Lambda}\mathbf{s}(n)
\end{equation}
where $\mathbf{P}$ is a permutation matrix and $\mathbf{\Lambda}$ is a non-singular diagonal matrix.

In the sequel, we propose to use the well known CMA to achieve the desired BSS. In other words, we propose to estimate the separation matrix by minimizing the CM criterion:
\begin{equation}\label{Eq04}
    \mathcal{J}(\mathbf{W})=\sum_{j=1}^{K}\sum_{i=1}^{M} \left(|z_{ij}|^{2}-1\right)^{2}
\end{equation}
where $z_{ij}$ is the $(i,j)$th entry of $\mathbf{Z}=\mathbf{W} \mathbf{Y}$, with $\mathbf{Y}=[\mathbf{y}(1), \mathbf{y}(2), \ldots, \mathbf{y}(K)]$ ($K$ being the sample size). This CM criterion has been used by many authors and has been shown to lead to the desired source separation for CM signals\footnote{In fact, the CMA can be used for sub-Gaussian sources (not necessary of constant modulus) as proved in \cite{Regalia_99}. %The equivalence between constant modulus and higher order statistics based criteria was investigated in \cite{Regalia_99}
} and large sample sizes as stated below.

\begin{theorem} If $K$ is large enough such that columns of matrix $\mathbf{S}=[\mathbf{s}(1), \mathbf{s}(2),\\ \ldots, \mathbf{s}(K)]$ include all possible combinations of source vectors\footnote{Note that this is a sufficient condition only.}
$\mathbf{s}(n)$, then the criterion $\mathcal{J}(\mathbf{W})$ (where $\mathbf{W}$ is such that $\mathbf{WA}$ is non singular) is minimized if and only if $\mathbf{W}$ satisfies:
\begin{equation}\label{Eq05}
\mathbf{W} \mathbf{A} = \mathbf{P} \mathbf{\Lambda}
\end{equation}
or, in the absence of noise:
\begin{equation}\label{Eq06}
     \mathbf{W}\mathbf{Y}=\mathbf{P}\mathbf{\Lambda}\mathbf{S}
\end{equation}
where $\mathbf{P}$ is an $M\times M$ permutation matrix and $\mathbf{\Lambda}$ is an $M\times M$ diagonal non-singular matrix.
\end{theorem}
\begin{proof}
   The proof can easily be derived from that of Theorem $3.2$ in \cite{Talwar_96}.
\end{proof}
%========================================================================================================================
%========================================================================================================================
%========================================================================================================================
%========================================================================================================================
\section{Givens CMA (G-CMA)}
\label{Sec:GCMA}
In this section, we propose a new algorithm, referred to as G-CMA, based on Givens rotations, for the minimization of the CM criterion\footnote{Part of this section's work has been presented in \cite{Ikhlef_10}.}. It is made up of two stages:
\begin{enumerate}
\item \textit{Prewhitening:} the prewhitening stage allows to convert the arbitrary channel matrix into a unitary one. Hence, this reduces finding an arbitrary separation matrix to finding a unitary one \cite{Comon_Bk}. Moreover, prewhitening has the advantage of reducing vector size (data compression) in the case where $N>M$ and avoiding trivial undesired solutions.
\item \textit{Givens rotations:} After prewhitening, the new  channel matrix is unitary and can therefore be computed via successive Givens rotations. Here, we propose to compute the optimal parameters of these rotations through minimizing the CM criterion.
\end{enumerate}

The prewhitening matrix $\mathbf{B}$ can be computed by using the classical eigendecomposition of the covariance matrix of the received signal $\mathbf{Y}$ (often, it is computed as the inverse square root of the data covariance matrix, $\frac{1}{K}\mathbf{Y}\mathbf{Y}^{H}$ \cite{Comon_Bk}). The whitened signal can then be written as:
\begin{equation}\label{Eq07}
    \bar{\mathbf{Y}}=\mathbf{B}\mathbf{Y}
\end{equation}
Therefore, assuming the noise free case and that the prewhitening matrix $\mathbf{B}$ is computed using the exact covariance matrix, we have:
\begin{equation}\label{Eq08}
    \bar{\mathbf{Y}}=\mathbf{B}\mathbf{A}\mathbf{S}=\mathbf{V}^{H}\mathbf{S}
\end{equation}
where $\mathbf{V}=\mathbf{A}^{H}\mathbf{B}^{H}$ is an $M\times M$ unitary matrix. From (\ref{Eq08}), it is clear that, in order to find the source signals, it is sufficient to find the unitary matrix $\mathbf{V}$ and hence the separator can simply be expressed as: $\mathbf{W} = \mathbf{V} \mathbf{B}$, which, in the absence of noise, results in $\mathbf{Z} = \mathbf{W} \mathbf{Y} = \mathbf{V} \mathbf{B} \mathbf{Y} = \mathbf{V} \bar{\mathbf{Y}} = \mathbf{V} \mathbf{V}^{H}\mathbf{S} = \mathbf{S}$.

Now, to minimize the CM criterion in (\ref{Eq04}) w.r.t. to matrix $\mathbf{V}$, we propose an iterative algorithm where $\mathbf{V}$ is rewritten using Givens rotations. Indeed, in Jacobi-like algorithms \cite{Golub_Bk_96}, the unitary matrix $\mathbf{V}$ can be decomposed into product of elementary complex Givens rotations $\mathbf{\Psi}_{pq}$ such that:
\begin{equation}\label{Eq09}
    \mathbf{V}=\prod_{N_{Sweeps}}~\prod_{1\leq p <q\leq M}\mathbf{\Psi}_{pq}
\end{equation}
where $N_{Sweeps}$ refers to the number of sweeps (iterations\footnote{In this paper we will use the terms \textit{iteration} and \textit{sweep} interchangeably.}) and the Givens rotation matrix $\mathbf{\Psi}_{pq}$ is a unitary matrix where all diagonal elements are one except for two elements $\psi_{pp}$ and $\psi_{qq}$. Likewise, all off-diagonal elements of $\mathbf{\Psi}_{pq}$ are zero except for two elements $\psi_{pq}$ and $\psi_{qp}$. Elements $\psi_{pp}, \psi_{pq}, \psi_{qp}$, and $\psi_{qq}$ are given by:
\begin{eqnarray} \label{Eq10}
\left[\begin{array}{cc}\psi_{pp} & \psi_{pq} \\ \psi_{qp} & \psi_{qq} \end{array} \right]
&=& \left[\begin{array}{cc} \cos (\theta) & e^{\jmath \alpha}  \sin(\theta)
\\-e^{-\jmath \alpha}  \sin(\theta) & \cos(\theta)\end{array}\right]
\end{eqnarray}

To compute $\mathbf{\Psi}_{pq}$, we need to find only the rotation angles $(\theta,\alpha)$. The idea here is to choose the rotation angles $(\theta,\alpha)$ such that the CM criterion $\mathcal{J}(\mathbf{V})$ is minimized. For this purpose, let us consider the unitary transformation\footnote{For simplicity, we keep using notation $\bar{\mathbf{Y}}$ even though the latter matrix is transformed at each iteration of the proposed algorithm.} $\breve{\mathbf{Y}}=\mathbf{\Psi}_{pq}\bar{\mathbf{Y}}$. Given the structure of $\mathbf{\Psi}_{pq}$, this unitary transformation changes only the elements in rows $p$ and $q$ of $\bar{\mathbf{Y}}$ according to:
\begin{equation}\label{Eq11}
            \breve{y}_{pj}=\cos (\theta)\bar{y}_{pj}+e^{\jmath \alpha}  \sin(\theta)\bar{y}_{qj} \mbox{  and  }
            \breve{y}_{qj}=-e^{-\jmath \alpha}  \sin(\theta)\bar{y}_{pj}+\cos (\theta)\bar{y}_{qj}
\end{equation}
where $\bar{y}_{ij}$ refers to the $(i,j)$th entry of $\bar{\mathbf{Y}}$.

The algorithm consists of minimizing iteratively the criterion in (\ref{Eq04}) by applying  successive Givens rotations, with initialization of $\mathbf{V}=\mathbf{I}$. $\mathbf{\Psi}_{pq}$ are computed such that $\mathcal{J}(\mathbf{\Psi}_{pq})$ is minimized at each iteration. In order to minimize $\mathcal{J}(\mathbf{\Psi}_{pq})$, we propose to express it as a function of $(\theta,\alpha)$. Since the application of Givens rotation matrix $\mathbf{\Psi}_{pq}$ to $\bar{\mathbf{Y}}$ modifies only the two rows $p$ and $q$, the terms that depend on $(\theta,\alpha)$ are those corresponding to $i=p$ or $i=q$ in (\ref{Eq04}). Considering (\ref{Eq10}) and
(\ref{Eq11}), we have:
\begin{eqnarray}\label{Eq12}
        \begin{array}{l}
    \mathcal{J}(\mathbf{\Psi}_{pq})=\sum_{j=1}^{K}\left[\big(|\breve{y}_{pj}|^{2}
    -1\big)^{2}+\big(|\breve{y}_{qj}|^{2}-1\big)^{2}\right] +\sum_{j=1}^{K}\sum_{i=1, i\neq p,q}^{M}
    \big(|\bar{y}_{ij}|^{2}-1\big)^{2}
            \end{array}
\end{eqnarray}
On the other hand, by considering (\ref{Eq11}) and the following equalities:
\begin{eqnarray}\label{Eq13}
	\begin{array}{l}
  		\cos^{2}(\theta) = \frac{1}{2}(1+\cos(2\theta)), 
  		\sin^{2}(\theta)=\frac{1}{2}(1-\cos(2\theta)), 
		\sin(2\theta) = 2\sin(\theta)\cos(\theta)
  	\end{array}
\end{eqnarray}
and after some manipulations, we obtain:
\begin{eqnarray}\label{Eq14}
\begin{split}
    |\breve{y}_{pj}|^{2} = \mathbf{t}_{j}^{T}\mathbf{v}+\frac{1}{2}\big(|\bar{y}_{pj}|^{2} +|\bar{y}_{qj}|^{2}\big)
    \mbox{ and }
    |\breve{y}_{qj}|^{2} =-\mathbf{t}_{j}^{T}\mathbf{v}+\frac{1}{2}\big(|\bar{y}_{pj}|^{2} +|\bar{y}_{qj}|^{2}\big)
\end{split}
\end{eqnarray}
with:
\begin{eqnarray}\label{Eq16}
    &&\mathbf{v}=[\cos(2\theta),~\sin(2\theta)\cos(\alpha),~\sin(2\theta)\sin(\alpha)]^{T}\label{Eq16}\\
    &&\mathbf{t}_{j}=\Big[\frac{1}{2}\big(|\bar{y}_{pj}|^{2}-|\bar{y}_{qj}|^{2}\big),~
     \Re(\bar{y}_{pj}\bar{y}_{qj}^{*}),~\Im(\bar{y}_{pj}\bar{y}_{qj}^{*})\Big]^{T}\label{Eq17}
\end{eqnarray}
where $\Re(a)$ and $\Im(a)$ denote real and imaginary parts of $a$, respectively. Using (\ref{Eq14}), we get:
\begin{align}\label{Eq18}
    \big(|\breve{y}_{pj}|^{2}-1\big)^{2}&+\big(|\breve{y}_{qj}|^{2}-1\big)^{2} = 2\mathbf{v}^{T} \mathbf{t}_{j} \mathbf{t}_{j}^{T}\mathbf{v}+2\left(\frac{|\bar{y}_{pj}|^{2} + \bar{y}_{qj}|^{2}}{2}-1\right)^{2}
\end{align}
Then, plugging  (\ref{Eq18}) into (\ref{Eq12}) yields:
\begin{eqnarray}\label{Eq19}
    \mathcal{J}(\mathbf{\Psi}_{pq}) &=& 2\sum_{j=1}^{K} \mathbf{v}^{T}\mathbf{t}_{j}\mathbf{t}_{j}^{T}\mathbf{v}
    +2\sum_{j=1}^{K}\left(\frac{|\bar{y}_{pj}|^{2}+|\bar{y}_{qj}|^{2}}{2}-1\right)^{2} \nonumber \\ 
    &+& \sum_{j=1}^{K}\sum_{i=1 \atop i\neq p,q}^{M} \big(|\bar{y}_{ij}|^{2}-1\big)^{2}
\end{eqnarray}
Given that the second and third summations in (\ref{Eq19}) do not depend on $(\theta,\alpha)$, the minimization problem is equivalent to the minimization of:
\begin{equation}\label{Eq20}
    \mathcal{F}(\mathbf{\Psi}_{pq})=\mathbf{v}^{T}\mathbf{T}\mathbf{v}
\end{equation}
where $\mathbf{T}=\sum_{j=1}^{K}\mathbf{t}_{j}\mathbf{t}_{j}^{T}$ and $\|\mathbf{v}\|=1$. Finally, the solution $\mathbf{v}$ that minimizes (\ref{Eq20}) is given by the unit norm eigenvector of $\mathbf{T}$ corresponding to the smallest eigenvalue\footnote{This is a $3\times3$ eigenvalue problem that can be solved explicitly.}. Given $\mathbf{v}=[v_{1},v_{2},v_{3}]^T$ we have:
\begin{eqnarray}\label{Eq21}
\begin{split}
\cos(\theta)=\sqrt{\frac{1+v_{1}}{2}}\mbox{ and } e^{\jmath \alpha}  \sin(\theta)=\frac{v_{2}+\jmath v_{3}}{\sqrt{2(1+v_{1})}}
\end{split}
\end{eqnarray}
Using (\ref{Eq21}), the computation of $\mathbf{\Psi}_{pq}$ follows directly from (\ref{Eq10}). The G-CMA algorithm is summarized in Table \ref{Tab:GCMA} (for simplicity, we use the same notation for the data and its transformed version).

%-------------------------------------------------------------------------------------------------------------------
\begin{table}[tb]
\renewcommand{\arraystretch}{1.5}
\centering
\begin{tabular}{l}%{|*{1}l |}
\hline Initialization: $\mathbf{V}=\mathbf{I}$\\
1.~~Prewhitening: $\bar{\mathbf{Y}}=\mathbf{B}\mathbf{Y}$, where $\mathbf{B}$ is the prewhitening matrix.\\
2.~~Complex Givens rotations:\\
~~~~~~~~\textbf{for} $i=1:N_{Sweeps}$\\
~~~~~~~~~~~~~~\textbf{for} $p=1:M-1$\\
~~~~~~~~~~~~~~~~~~~~\textbf{for} $q=p+1:M$\\
~~~~~~~~~~~~~~~~~~~~~~~~~~Compute $\mathbf{\Psi}_{pq}$ using (\ref{Eq21})\\
%$\mathbf{\Psi}_{pq}=\underset{\theta,\alpha}{\operatorname{argmin}} \quad\mathcal{J}'(\mathbf{\Psi}_{pq})$\\
~~~~~~~~~~~~~~~~~~~~~~~~~~$\bar{\mathbf{Y}}=\mathbf{\Psi}_{pq}\bar{\mathbf{Y}}$\\
~~~~~~~~~~~~~~~~~~~~~~~~~~$\mathbf{V}=\mathbf{\Psi}_{pq}\mathbf{V}$\\
~~~~~~~~~~~~~~~~~~~~\textbf{end for}\\
~~~~~~~~~~~~~~\textbf{end for}\\
~~~~~~~~\textbf{end for}\\
3.~~After convergence, computation of the separation matrix: $\mathbf{W}=\mathbf{V}\mathbf{B}$\\
4.~~Separation: $\hat{\mathbf{S}}=\mathbf{W}\mathbf{Y}=\bar{\mathbf{Y}}$.\\
\hline
\end{tabular}
\caption{The Givens CMA (G-CMA) algorithm.} \label{Tab:GCMA}
\end{table}
%-------------------------------------------------------------------------------------------------------------------
The G-CMA algorithm described above requires that the number of samples available at the receiver is large enough so that the prewhitening step results in an equivalent channel matrix close to unitary, for which the use of Givens rotations is effective. However, for small numbers of samples, prewhitening may result in an equivalent channel matrix not close to unitary, in which case, applying G-CMA alone is ineffective. Next, we propose to solve this problem by introducing the Hyperbolic Givens rotations.
%========================================================================================================================
%========================================================================================================================
\section{Hyperbolic Givens CMA (HG-CMA)}
\label{Sec:HGCMA}
As stated in the previous section, the use of Givens rotations in the case of small numbers of samples is not effective. To overcome this limitation, we introduce here the use of Hyperbolic Givens rotations. The latter consist of applying Shear rotations and Givens rotations alternatively. Matrix $\mathbf{W}$ can be decomposed into product of elementary complex Shear rotations, Givens rotations and normalization transformation as follows:
\begin{equation} \label{Eq23}
  \mathbf{W}= \prod_{N_{Sweeps}}~~\prod_{1\leq p<q \leq M} \mathbf{D}_{pq}~\mathbf{\Psi}_{pq}~\mathbf{H}_{pq}
\end{equation}
where $\mathbf{D}_{pq}$, $\mathbf{\Psi}_{pq}$ and $\mathbf{H}_{pq}$ denote normalization, unitary Givens and non-unitary Shear transformations, respectively. The unitary matrix $\mathbf{\Psi}_{pq}$ is defined in (\ref{Eq10}). Similar to $\mathbf{\Psi}_{pq}$, $\mathbf{H}_{pq}$ is equal to the identity matrix except for the elements $h_{pp}, h_{pq}, h_{qp}$ and $h_{qq}$ that are given by:
\begin{eqnarray} \label{Eq24}
\left[\begin{array}{cc}h_{pp} & h_{pq} \\ h_{qp} & h_{qq} \end{array} \right] &=&
\left[\begin{array}{cc}\cosh (\gamma) & e^{\jmath \beta} \sinh(\gamma)
\\e^{-\jmath \beta}  \sinh(\gamma) & \cosh(\gamma)\end{array}\right]
\end{eqnarray}
where $ \gamma \in \mathbb{R}$ is the hyperbolic transformation parameter and $ \beta \in [-\frac{\pi}{2} , \frac{\pi}{2}]$ is an angle parameter (equal to zero in the real case).
The normalization transformation $\mathbf{D}_{pq}=\mathbf{D}_{pq}(\lambda_{p},\lambda_{q})$ is a diagonal matrix with diagonal elements  equal to one except for the two elements $d_{pp}=\lambda_p$, and $d_{qq}=\lambda_q$.

In the following derivation, we consider the square case where $N=M$ (if $N > M$, one can use signal subspace projection as in \cite{Veen_Chap_05}).
%========================================================================================================================
%========================================================================================================================
\subsection{Non-Unitary Shear Rotations}
\label{Sub:Hyperbolic}
By applying $\mathbf{H}_{pq}$ to the received signal, we get:
\begin{equation} \label{Eq25}
   \tilde{\mathbf{Y}} = \mathbf{H}_{pq}~\mathbf{Y}
\end{equation}

From (\ref{Eq24}), only the $p$th and $q$th rows of $\mathbf{Y}$ are affected according to:
\begin{eqnarray}\label{Eq26}
\begin{split}
\tilde{y}_{pj}&=& \cosh(\gamma) {y_{pj}} + e^{\jmath \beta} \sinh(\gamma) {y_{qj}} \mbox{ and }
\tilde{y}_{qj}&=&  e^{-\jmath \beta} \sinh(\gamma) {y_{pj}} + \cosh(\gamma) {y_{qj}}
\end{split}
\end{eqnarray}
In order to compute $\mathbf{H}_{pq}$, we propose to minimize the CM cost function in (\ref{Eq04}) w.r.t. $\mathbf{H}_{pq}$:
\begin{eqnarray}\label{Eq27}
 \mathcal{J}(\mathbf{H}_{pq})= \sum_{j=1}^{K} (|\tilde{y}_{pj}|^2-1)^2 + (|\tilde{y}_{qj}|^2-1)^2 + \sum_{j=1}^{K} \sum_{i=1 \atop i\neq p,q}^{M}({|\bar{y}_{ij}|^2}-1)^2
\end{eqnarray}
By considering (\ref{Eq26}) and the following equalities: 
\begin{eqnarray}\label{Eq28}
\begin{split}
\sinh(2 \gamma) = 2 \sinh(\gamma) \cosh(\gamma), 
\cosh^2(\gamma) = \frac{1}{2} (\cosh(2 \gamma) + 1), 
\sinh^2(\gamma) = \frac{1}{2} (\cosh(2 \gamma) - 1)
\end{split}
\end{eqnarray}
and after some straightforward derivations, we obtain:
\begin{eqnarray}\label{Eq29}
      |\tilde{y}_{pj}|^2 =  \mathbf{r}_{j}^T \mathbf{u} + \frac {1}{2} (|y_{pj}|^2 - |{y_{qj}|^2}) \mbox{ and }
      |\tilde{y}_{qj}|^2 =  \mathbf{r}_{j}^T \mathbf{u} - \frac {1}{2} (|y_{pj}|^2 - |{y_{qj}|^2})
\end{eqnarray}
with:
\begin{eqnarray}
&&\mathbf{u}=\left[\cosh(2\gamma),\;\;\cos(\beta)\;\sinh(2\gamma),\;\;\sin(\beta)\;\sinh(2\gamma)\right]^T \label{Eq30} \\ 
&&\mathbf{r}_{j}=\left[\frac{1}{2}\left(|y_{pj}|^2+|{y_{qj}|^2}\right),\;\;\Re\left(y_{pj}y_{qj}^*\right),\;\;\Im\left(y_{pj}y_{qj}^*\right)\right]^T\label{Eq31}
\end{eqnarray}
Using the results in (\ref{Eq29}), we can rewrite the first two terms in (\ref{Eq27}) as:
\begin{eqnarray}\label{Eq32}
\left(|\tilde{y}_{pj}|^2-1\right)^2 + \left(|\tilde{y}_{qj}|^2-1\right)^2 &=& 2 \mathbf{u}^T \mathbf{r}_{j} \mathbf{r}_{j}^T \mathbf{u} - 4 \mathbf{u}^T \mathbf{r}_{j} \nonumber \\
&+& \frac {1}{2}\left(|\bar{y}_{pj}|^2 - |{\bar{y}_{qj}|^2}\right)^2 + 2
\end{eqnarray}
Then, by substituting (\ref{Eq32}) into (\ref{Eq27}), we obtain:
\begin{eqnarray} \label{Eq33}
\mathcal{J}(\mathbf{u}) = 2\left(\sum_{j=1}^{K} \mathbf{u}^T\mathbf{r}_{j}\mathbf{r}_{j}^T\mathbf{u}-2 \mathbf{u}^T \mathbf{r}_{j}\right) &+& 2 \sum_{j=1}^{K} \big[\frac {1}{4}({|\bar{y}_{pj}|^2} - {|\bar{y}_{qj}|^2})^2 + 1\big] \nonumber \\ &+& \sum_{j=1}^{K} \sum_{i=1 \atop i\neq p,q}^{M}({|\bar{y}_{ij}|^2}-1)^2
\end{eqnarray}
We note that only the first term on the right hand side of the equality (\ref{Eq33}) depends on $(\gamma,\beta)$, and hence the minimization of (\ref{Eq33}) is equivalent to the minimization of:
\begin{equation}\label{Eq34}
\mathcal{F}(\mathbf{u})=\sum_{j=1}^{K}\mathbf{u}^T\mathbf{r}_{j}\mathbf{r}_{j}^T\mathbf{u}-2\mathbf{u}^T \mathbf{r}_{j}
\end{equation}
This optimisation problem can be achieved in three different ways: by computing the exact solution, by taking linear approximation to zero, and with semi linear approximation.
%========================================================================================================================
%========================================================================================================================
\subsubsection{Exact Solution}
\label{Sub:ExactSol}
In this approach, we compute the optimum solution using the Lagrange multiplier method. The optimization problem can be expressed as:
\begin{eqnarray} \label{Eq35}
    \min_{\mathbf{u}}~~\mathcal{F}(\mathbf{u})~~~\textrm{s.t.}~~~\mathbf{u}^T \mathbf{J}_{3}\mathbf{u} = 1
\end{eqnarray}
where $\mathbf{J}_{3} = \mbox{diag}\left(\left[1,-1,~-1\right]\right)$ so that constraint is equivalent to $\cosh^2(2\gamma)-\sinh^2(2\gamma)=1$. The Lagrangian of the optimization problem in (\ref{Eq35}) can be written as:
\begin{equation} \label{Eq36}
\mathcal{L}(\mathbf{u},\lambda) = \mathbf{u}^T \mathbf{R} \mathbf{u}-2 \mathbf{r}^T \mathbf{u}+\lambda(\mathbf{u}^T \mathbf{J}_{3} \mathbf{u}-1)
\end{equation}
where $\mathbf{R}=\sum_{j=1}^K\mathbf{r}_{j}\mathbf{r}_{j}^T$ is a $(3\times3)$ symmetric matrix, $\mathbf{r}=\sum_{j=1}^K\mathbf{r}_{j}$, $\mathbf{u}$ and $\mathbf{r}_{j}$ are defined in (\ref{Eq30}) and (\ref{Eq31}), respectively. The solution that minimizes the Lagrangian in (\ref{Eq36}) can be expressed as:
\begin{equation}\label{Eq38}
    \mathbf{u}=(\mathbf{R}+\lambda \mathbf{J}_{3})^{-1} \mathbf{r}
\end{equation}
where $\lambda$ is the solution of:
\begin{equation}\label{Eq39}
\mathbf{u}^T\mathbf{J}_{3}\mathbf{u}=1\Longleftrightarrow \mathbf{r}^T(\mathbf{R}+\lambda \mathbf{J}_{3})^{-1} \mathbf{J}_{3} (\mathbf{R}+\lambda \mathbf{J}_{3})^{-1} \mathbf{r}=1
\end{equation}
which is a $6$-th order polynomial equation (see appendix A) of the form: $P_6(\lambda)=c_0\lambda^6+c_1\lambda^5+c_2\lambda^4+c_3\lambda^3+c_4\lambda^2+c_5\lambda+c_6=0$.
The desired solution $\lambda$ is the real-valued root of the above polynomial that corresponds to the minimum value of (\ref{Eq36}). Finally, given the solution $\mathbf{u}=[u_1\;u_2\;u_3]^T$ in (\ref{Eq38}), the Shear transformation entries are computed as:
\begin{equation} \label{Eq40}
h_{pp} = h_{qq} =  \sqrt{\frac{u_1 +1}{2}}\mbox{ and } h_{pq} = h_{qp}^* = \frac{(u_2+\jmath u_3)}{2 h_{pp}}
\end{equation}

Note that, for the computation of each Shear rotation matrix, we need to perform a $3 \times 3$ matrix inversion and solve a $6$-th order polynomial equation. Hence, as the number of sweeps and transmit antennas increases, the complexity increases. In the following, we present two suboptimal solutions that have less complexity and close performance compared to the exact one.
%========================================================================================================================
%========================================================================================================================
\subsubsection{Semi-Exact Solution}
\label{Sub:SemiExactSol}
We denote this approach by \textit{semi-exact solution}, since for computing $\beta$ we take the approximation in (\ref{Eq41}), while for the angle rotation $\gamma$ we compute an exact solution using the Lagrange multiplier method. By considering the first order approximation around zero of $\sinh$ and $\cosh$, we have:
\begin{eqnarray}\label{Eq41}
\sinh(2\gamma) \approx 2\sinh(\gamma) \approx 2\gamma \mbox{ and } \cosh(2\gamma) \approx \cosh(\gamma) \approx 1
\end{eqnarray}
Using (\ref{Eq41}) in (\ref{Eq30}), equation (\ref{Eq34}) can be expressed as:
\begin{eqnarray}\label{Eq42}
\mathcal{F}(\gamma,\beta)=\sum_{j=1}^{K} r_{j}^{(1)} \left(r_{j}^{(1)}-2\right) &+& 4\gamma \left[\cos(\beta) r_{j}^{(2)}\left(r_{j}^{(1)}-1\right)+\sin(\beta) r_{j}^{(3)}\left(r_{j}^{(1)}-1\right)\right] \nonumber \\ 
&+& 4\gamma^2\left(\cos(\beta)r_{j}^{(2)}+\sin(\beta)r_{j}^{(3)}\right)^2
\end{eqnarray}
where $r_{j}^{(i)}$ is the $i$th element of $\mathbf{r}_{j}$. The linear approximation of (\ref{Eq42}) for $\gamma$ close to zero (which corresponds to simply neglecting the terms involving $\gamma^n$ for $n\geq2$) can be obtained by discarding the last term of (\ref{Eq42}):
\begin{equation}\label{Eq43}
\begin{array}{l}
\mathcal{F}(\gamma,\beta)\approx \sum_{j=1}^{K} r_{j}^{(1)}\left(r_{j}^{(1)}-2\right) + 4\gamma\left[\cos(\beta)r_{j}^{(2)} \left(r_{j}^{(1)}-1\right)+\sin(\beta)r_{j}^{(3)}\left(r_{j}^{(1)}-1\right)\right]
\end{array}
\end{equation}
The minimization of (\ref{Eq43}) obtained by zeroing its derivative) leads to:
\begin{equation}\label{Eq44}
\begin{array}{l}
\beta = \mathrm{arctan} \left(\frac{\sum_{j=1}^K r_{j}^{(3)}\; \left(r_{j}^{(1)}-1\right)}{\sum_{j=1}^K r_{j}^{(2)}\; \left(r_{j}^{(1)}-1\right)} \right)
\end{array}
\end{equation}
Once we have $\beta$, let us define:
\begin{eqnarray}
&&\tilde{\mathbf{u}} =[\cosh(2\gamma),\;\sinh(2\gamma)]^T\label{Eq45}\\ 
&&\tilde{\mathbf{r}}_{j}=\left[\frac{1}{2}\left(|y_{pj}|^2+|{y_{qj}|^2}\right),~\cos(\beta) \Re(y_{pj}y_{qj}^*) + \sin(\beta) \Im(y_{pj}y_{qj}^*) \right]^T\label{Eq46}
\end{eqnarray}
and hence, finding $\gamma$ which minimizes (\ref{Eq34}) implies solving the following optimization problem:
\begin{equation} \label{Eq47}
\min_{\tilde{\mathbf{u}}}~~\mathcal{K}(\tilde{\mathbf{u}})~~~\textrm{s.t.}~~~\tilde{\mathbf{u}}^T\mathbf{J}_{2}\tilde{\mathbf{u}} = 1
\end{equation}
where $\mathbf{J}_{2} = \mbox{diag}\left([1,~-1]\right)$ and:
\begin{eqnarray}\label{Eq48}
\begin{array}{l}
\mathcal{K}(\tilde{\mathbf{u}}) = \sum_{j=1}^{K}\tilde{\mathbf{u}}^T \tilde{\mathbf{r}}_{j}\tilde{\mathbf{r}}_{j}^T \tilde{\mathbf{u}}-2 \tilde{\mathbf{u}}^T \tilde{\mathbf{r}}_{j}
\end{array}
\end{eqnarray} 
By defining $\tilde{\mathbf{R}}=\sum_{j=1}^{K}\tilde{\mathbf{r}}_{j}\tilde{\mathbf{r}}_{j}^T$ and $\tilde{\mathbf{r}}=\sum_{j=1}^K\tilde{\mathbf{r}}_{j}$, the optimization of (\ref{Eq47}) using Lagrange multiplier leads to:
\begin{equation}\label{Eq50}
    \tilde{\mathbf{u}}=(\tilde{\mathbf{R}}+\lambda \mathbf{J}_{2})^{-1} \tilde{\mathbf{r}}
\end{equation}
where $\lambda$ is the solution of:
\begin{equation}\label{Eq51}
    \tilde{\mathbf{u}}^T \mathbf{J}_{2} \tilde{\mathbf{u}} = 1 \Longleftrightarrow \tilde{\mathbf{r}}^T (\tilde{\mathbf{R}}+\lambda \mathbf{J}_{2})^{-1}\mathbf{J}_{2}(\tilde{\mathbf{R}}+\lambda \mathbf{J})^{-1}\tilde{\mathbf{r}} = 1
\end{equation}
This is a $4$-th order polynomial equation (see appendix A) of the form: $P_4(\lambda)=c_0\lambda^4+c_1\lambda^3+c_2\lambda^2+c_3\lambda+c_4= 0$. The desired solution $\lambda$ is the real-valued root of the above polynomial that corresponds to the minimum value of (\ref{Eq48}). Finally, given the solution $\tilde{\mathbf{u}}=[\tilde{u}_1\;\tilde{u}_2]^T$ in (\ref{Eq50}) and $\beta$ in (\ref{Eq44}), the Shear transformation entries can be obtained as:
\begin{eqnarray}\label{Eq52}
        h_{pp} & = h_{qq} =  \sqrt{\frac{1}{2}(\tilde{u}_1 +1)}\mbox{ and }
        h_{pq} & = h_{qp}^* = e^{\jmath \beta}\frac{\tilde{u}_2}{2 h_{pp}}
\end{eqnarray}
We note that in this solution, for the computation of each Shear rotation matrix, we need to solve a $4$-th order polynomial equation. Hence, the complexity of this solution is clearly less than that of the exact one.
%========================================================================================================================
%========================================================================================================================
\subsubsection{Solution with Linear Approximation to Zero}
\label{Sub:ApproSol}
In this approach, we compute $\beta$ as in (\ref{Eq44}) and then we compute $\gamma$ which minimizes (\ref{Eq48}) by considering the approximation in (\ref{Eq41}). We define:
\begin{equation} \label{Eq53}
\tilde{\mathbf{R}} = \left[\begin{array}{cc}\tilde{r}_{11} &\tilde{r}_{12}\\ \tilde{r}_{21} &\tilde{r}_{22} \end{array}\right] ~~ \mathrm{and} ~~\tilde{\mathbf{r}} = \left[\begin{array}{cc}\tilde{r}_{1}\\ \tilde{r}_{2}\end{array}\right]
\end{equation}
and using (\ref{Eq26}), (\ref{Eq48}) can be written as:
\begin{equation} \label{Eq54}
    \mathcal{K}(\gamma)=\frac{1}{2}(\tilde{r}_{11}+ \tilde{r}_{22})\cosh(4\gamma)+\tilde{r}_{12} \sinh(4\gamma)-2 \tilde{r}_{1} \cosh(2\gamma)-2 \tilde{r}_{2}\sinh(2\gamma)
\end{equation}
By taking the first derivative of (\ref{Eq54}) with respect to $\gamma$, using (\ref{Eq41}), and setting the result equal to zero, we obtain:
\begin{equation} \label{Eq55}
    \sinh(2\gamma) (\tilde{r}_{11}+\tilde{r}_{22}-\tilde{r}_{1}) + \cosh(2\gamma)(\tilde{r}_{12}-\tilde{r}_{2}) = 0
\end{equation}
Which solution is:
\begin{eqnarray}\label{Eq56}
\begin{array}{l}
\gamma =\frac{1}{2}\mathrm{arctanh}\left(\frac{\sum_{j=1}^K \left[\left(\cos(\beta) r_{j}^{(2)}+\sin(\beta) r_{j}^{(3)}\right) \left(1-r_{j}^{(1)}\right)\right]}{\sum_{j=1}^K \left[\left((r_{j}^{(1)})^2-r_{j}^{(1)}\right)+\left(\cos(\beta) r_{j}^{(2)}+\sin(\beta) r_{j}^{(3)}\right)^2\right]} \right)
\end{array}
\end{eqnarray}
Given $\beta$ in (\ref{Eq44}) and $\gamma$ in (\ref{Eq56}), the computation of $\mathbf{H}_{pq}$ follows directly. This solution has the lowest complexity among the three considered ones.
%========================================================================================================================
%========================================================================================================================
\subsection{Unitary Givens Rotation}
\label{Sub:Givens}
After the Shear transformation, we now apply the Givens transformation to the result of the Shear rotation as:
\begin{equation} \label{Eq57}
\mathbf{\underbar{Y}} = \mathbf{\Psi}_{pq} \tilde{\mathbf{Y}}
\end{equation}
The unitary matrix $\mathbf{\Psi}_{pq}$ is computed in the same way as in Section \ref{Sec:GCMA}. 
%========================================================================================================================
%========================================================================================================================
%========================================================================================================================
\subsection{Normalization Rotations}
\label{Sub:HGCMA_Norm}

%\underline{Remark}: In our CM criterion, we have set the constant equal to one while in the original CM criterion it is chosen equal to $E(\absF{s_i})/E(\absT{s_i})$. Somehow, this normalization step is introduced to compensate for this constant choice (a value which is supposed unknown in a blind context).

The last algorithm's transform is a normalization step. In our CM criterion in (\ref{Eq04}), we have set the constant equal to one while in the original CM criterion it is chosen equal to $C_i = E[\absF{s_i}]/E[\absT{s_i}]$. Somehow, this normalization step is introduced to compensate for this constant choice (the value  of $C_i$ is supposed unknown in a blind context). 

It has been shown in the two previous subsection that both Givens and hyperbolic transformations affect only the rows of indices $p$ and $q$ of the data bloc $\mathbf{\underbar{Y}}$ which means that only these two rows need to be normalized:
\begin{equation}\label{Eq58}
    \mathbf{Z}=\mathbf{D}_{(pq)}(\lambda_p,\lambda_q)~\mathbf{\underbar{Y}}
\end{equation}

The optimal parameters $(\lambda_p,\lambda_q)$ are calculated so that they minimize the CM criterion in (\ref{Eq04}) w.r.t. $\mathbf{D}_{(pq)}(\lambda_p,\lambda_q)$. The CM criterion is expressed in this case as (constant terms are omitted):
\begin{equation} \label{Eq59}	
\mathcal{J}_{D}(\lambda_p,\lambda_q)=\sum_{j=1}^{K}(\lambda_p^4\absF{\underbar{y}_{pj}}-2\lambda_p^2\absT{\underbar{y}_{pj}}) + \sum_{j=1}^{K}(\lambda_q^4\absF{\underbar{y}_{qj}}-2\lambda_q^2\absT{\underbar{y}_{qj}})
\end{equation}
Optimal normalization parameters can be obtained at the zeros of the derivatives of (\ref{Eq59}) with respect to these two parameters as follows:
\begin{eqnarray} \label{Eq60}
\begin{array}{lll}
\lambda_p = \sqrt{\sum_{j=1}^{K}\absT{\underbar{y}_{pj}} / \sum_{j=1}^{K}\absF{\underbar{y}_{pj}}}
~~\mbox{ and } \lambda_q = \sqrt{\sum_{j=1}^{K}\absT{\underbar{y}_{qj}} / \sum_{j=1}^{K}\absF{\underbar{y}_{qj}}} 
\end{array}
\end{eqnarray}
The HG-CMA algorithm is summarized in Table \ref{Tab:HGCMA}.

%Table HG-CMA Algorithm
\begin{table}[tb]
\renewcommand{\arraystretch}{1.5}
\centering
\begin{tabular}{l}%{|*{1}l |}
\hline Initialization: $\mathbf{W}=\mathbf{I}$\\
Signal subspace projection if $N>M$ \\
\textbf{for} $i=1:N_{Sweeps}$\\
~~~~~~\textbf{for} $p=1:M-1$\\
~~~~~~~~~~~~\textbf{for} $q=p+1:M$\\
~~~~~~~~~~~~~~~~~~Compute $\mathbf{H}_{pq}$:\\
~~~~~~~~~~~~~~~~~~~~~~- using (\ref{Eq40}) for exact solution\\
~~~~~~~~~~~~~~~~~~~~~~- using (\ref{Eq44}) and (\ref{Eq52}) for semi exact solution\\
~~~~~~~~~~~~~~~~~~~~~~- using (\ref{Eq44}) and (\ref{Eq56}) for linear approximation to zero (preferred)\\
~~~~~~~~~~~~~~~~~~$\mathbf{Y}=\mathbf{H}_{pq}\mathbf{Y}$\\
~~~~~~~~~~~~~~~~~~$\mathbf{W}=\mathbf{H}_{pq}\mathbf{W}$\\
~~~~~~~~~~~~~~~~~~Compute $\mathbf{\Psi}_{pq}$ using (\ref{Eq21})\\
~~~~~~~~~~~~~~~~~~$\mathbf{Y}=\mathbf{\Psi}_{pq}\mathbf{Y}$\\
~~~~~~~~~~~~~~~~~~$\mathbf{W}=\mathbf{\Psi}_{pq}\mathbf{W}$\\
~~~~~~~~~~~~~~~~~~Compute $\mathbf{D}_{pq}$ using (\ref{Eq60})\\
~~~~~~~~~~~~~~~~~~$\mathbf{Y}=\mathbf{D}_{pq}\mathbf{Y}$\\
~~~~~~~~~~~~~~~~~~$\mathbf{W}=\mathbf{D}_{pq}\mathbf{W}$\\
~~~~~~~~~~~~\textbf{end for}\\
~~~~~~\textbf{end for}\\
\textbf{end for}\\
Separation: $\hat{\mathbf{S}}=\mathbf{W}\mathbf{Y}=\mathbf{Y}$.\\
\hline
\end{tabular}
\caption{The Hyperbolic Givens CMA (HG-CMA) algorithm.} \label{Tab:HGCMA}
\end{table}
%========================================================================================================================
%========================================================================================================================
\section{Adaptive HG-CMA}
\label{Sec:AHGCMA}
To make an adaptive version of the HG-CMA algorithm, let us consider a sliding bloc of size $K$, $\mathbf{Y}^{(t-1)}=\left[\mathbf{y}(t-K),...,\mathbf{y}(t-2),\mathbf{y}(t-1)\right]$ which is updated at each new acquisition of a new sample $\mathbf{y}(t)$ (at time instant $t$). The main idea of the adaptive HG-CMA is to apply only one sweep of complex rotations on the sliding window at each time instant and update the separation matrix $\mathbf{W}$ by this sweep of rotations.

The numerical cost of the HG-CMA is of order $O(KM^2)$ (assuming $K >M$) but can be reduced to $O(KM)$ flops per iteration if we use only one or two rotations per time instant. In the simulation experiments, we compare the performance of the algorithm in the 3 following cases:
\begin{itemize}
\item When we use one complete sweep (i.e. $M(M-1)/2$ rotations)
\item When we use one single rotation which indices are chosen according to an automatic selection (i.e. automatic incrementation) throughout the iterations in such a way all search directions are visited periodically.
\item When we use two rotations per iteration (time instant): one pair of indices is selected according to the maximum deviation criterion:
\begin{equation}\label{Eq61}
\begin{array}{l}
(p,q) = arg\max \sum_{k=1}^K(|y_{pk}|^2 -1)^2+(|y_{qk}|^2 -1)^2
\end{array}
\end{equation}
the other rotation indices are selected automatically.
\end{itemize}
Comparatively, the adaptive ACMA \cite{Veen_Chap_05} costs approximately $O(M^3)$ flops per iteration and the LS-CMA\footnote{We consider here an adaptive version of the LS-CMA using the same sliding window as for our algorithm.} costs $O(KM^2 + M^3)$.  Interestingly, as shown in section \ref{Sec:Results}, the sliding window length $K$ can be chosen of the same order as the number of sources $M$ without affecting much the algorithm's performance. In that case, the numerical cost of HG-CMA becomes similar to that of the adaptive ACMA. The adaptive HG-CMA algorithm is summarized in Table \ref{Tab:AHGCMA}. Note that, the normalization step is done outside the sweep loop which reduces slightly the numerical cost.

%-------------------------------------------------------------------------------------------------------------------
\begin{table}[tb]
\centering
\renewcommand{\arraystretch}{1.5}
\begin{tabular}{l}%{|*{1}l |}
\hline
~Initialization:~$\mathbf{W}^{(K)}=\mathbf{I}_{M}$\\
~\textbf{For} $t = K+1, K+2, ...$ \textbf{do}\\
~~~~~~$\mathbf{y}(t)=\mathbf{W}^{(t-1)}~\mathbf{y}(t)$\\
~~~~~~$\mathbf{Y}^{(t)}=\left[\mathbf{y}(t-K),...,\mathbf{y}(t-1),\mathbf{y}(t)\right]$\\
~~~~~~$\mathbf{W}^{(t)}=\mathbf{W}^{(t-1)}$\\
~~~~~~\textbf{For all} $1 \leq p <q \leq M$ \textbf{do}\\
~~~~~~~~~~~~Compute $\mathbf{H}_{(pq)}$ using (\ref{Eq44}) and (\ref{Eq56})\\
~~~~~~~~~~~~Compute $\mathbf{\Psi}_{(pq)}$ using (\ref{Eq21})\\
~~~~~~~~~~~~Update  $\mathbf{W}^{(t)}=\mathbf{\Psi}_{(pq)}~\mathbf{H}_{(pq)}~\mathbf{W}^{(t)}$\\
~~~~~~~~~~~~Update  $\mathbf{Y}^{(t)}=\mathbf{\Psi}_{(pq)}~\mathbf{H}_{(pq)}~\mathbf{Y}^{(t)}$\\
~~~~~~\textbf{end For} \\
~~~~~~\textbf{For} $1\leq p\leq M$,~compute $\lambda_p$ using (\ref{Eq60}),~\textbf{end For}\\
~~~~~~Compute $\mathbf{D}=\mbox{diag}([\lambda_1, \cdots, \lambda_M])$\\
~~~~~~Update $\mathbf{W}^{(t)}=\mathbf{D}~\mathbf{W}^{(t)}$ and $\mathbf{Y}^{(t)}=\mathbf{D}~\mathbf{Y}^{(t)}$\\
~\textbf{end For}\\
\hline
\end{tabular}
\caption{Adaptive HG-CMA Algorithm.}
\label{Tab:AHGCMA}
\end{table}
%-------------------------------------------------------------------------------------------------------------------
%========================================================================================================================
%========================================================================================================================
%========================================================================================================================
\section{Numerical Results}
\label{Sec:Results}
Some numerical results are now presented in order to assess the performance of the proposed algorithms. For comparison we use ACMA \cite{Veen_ACMA_96} and LS-CMA \cite{Veen_Chap_05} as a benchmark. As performance measure, we use the signal to interference
and noise ratio (SINR) defined as:
\begin{eqnarray}\label{eq60}
\begin{array}{l}
\textrm{SINR}=\frac{1}{M}\sum_{k=1}^{M}\textrm{SINR}_{k}\mbox{ with }~~\textrm{SINR}_{k}= \frac{|g_{kk}|^{2}}{\sum\limits_{\ell,\ell\neq k}|g_{k\ell}|^{2}+\mathbf{w}_{k}\mathbf{R}_{b}\mathbf{w}_{k}^{H}}
\end{array}
\end{eqnarray}
where $\textrm{SINR}_{k}$ is the signal to interference and noise ratio at the $k$th output $g_{ij}=\mathbf{w}_{i}\mathbf{a}_{j}$, where $\mathbf{w}_{i}$ and $\mathbf{a}_{j}$ are the $i$th row vector and $j$th column vector of matrices $\mathbf{W}$ and $\mathbf{A}$, respectively.
$\mathbf{R}_{b}=E[\mathbf{b}\mathbf{b}^{H}]=\sigma_{b}^{2}\mathbf{I}_{N}$ is the noise covariance matrix. The source signals are assumed to be of unit variance.

We use the data model in (\ref{Eq01}); The system inputs are independent, uniformly distributed and drawn from 8-PSK, or 16-QAM constellations. The channel matrices $\mathbf{A}$ are generated randomly at each run but with controlled conditioning (their entries are generated as i.i.d. Gaussian variables). Unless otherwise specified, we consider $M=5$ transmit and $N=7$ receive antennas. The noise variance is determined according to the desired signal to noise ratio (SNR). In all figures the results are averaged over 1000 independent realizations (Monte Carlo runs).

Fig. \ref{Fig_1} depicts the SINR of HG-CMA vs. the SNR. We compare the three solutions, i.e., linear approximation to zero, semi-exact and exact solutions for Shear rotations in HG-CMA for 8-PSK and 16-QAM constellations. The sample size is $K=100$ and the number of sweeps is set equal to 10. We observe that the three solutions have almost the same performance for both 8-PSK and 16-QAM constellations. Therefore, in the following simulations, in HG-CMA, we will consider the linear approximation to zero solution.

In Fig. \ref{Fig_2}, we  investigate the effect of the number of sweeps on the performance of G-CMA and HG-CMA. The figure shows the SINR vs. the SNR for different numbers of sweeps. In this simulation, we assumed 8-PSK constellation and $K=100$ samples. We observe that, as expected, the performance is improved by increasing the number of sweeps and from 5 sweeps upwards, the performance remains unchangeable. In the rest of this section we consider $10$ sweeps in G-CMA and HG-CMA. Moreover, we can see that for small number of iterations HG-CMA is much better than G-CMA and the gap between them decreases as the number of iterations increases.

Fig. \ref{Fig_3} compares the proposed HG-CMA and G-CMA algorithms with ACMA in terms of SINR vs. SNR for 8-PSK constellation and various numbers of samples. We observe that, as expected, the larger the number of samples, the better the performance for all algorithms. For small number of samples, i.e. $K=20$, we observe that HG-CMA significantly outperforms ACMA and G-CMA. We also observe that G-CMA performs better than ACMA  for low to moderate SNR while for $\mathrm{SNR}>23~\mathrm{dB}$, ACMA becomes better. The reason that ACMA performs worse than HG-CMA is that the number of samples $K=20$ is less than the number of transmit antennas squared $M^2$, i.e., $K=20<M^2=25$ and as we stated above for ACMA to achieve good performance in the case of PSK constellations the number of samples $K$ must be at least greater than $M^2$ \cite{Veen_01}. For $K=100$, HG-CMA still provides the best performance while the performance of ACMA becomes very close to that of HG-CMA and better than that of G-CMA. We can say that for small or moderate number of samples the proposed algorithms are more suitable as compared to ACMA even for PSK constellations.

In Fig. \ref{Fig_4}, we consider the case of 16-QAM constellation. We notice that the proposed HG-CMA and G-CMA algorithms provide better performance as compared to ACMA. We also observe that, unlike the 8-PSK case in Fig. \ref{Fig_3}, the performance of HG-CMA and G-CMA are close in the case of 16-QAM. Moreover, we can see that the gap between the performance of the proposed algorithms and ACMA gets smaller as the number of samples $K$ increases. We can say that the proposed HG-CMA and G-CMA algorithms are more suitable as compared to ACMA for non-constant modulus constellations, since they provide better performance for a lower computational cost.

In Figs. \ref{Fig_5} and \ref{Fig_6}, we plot the SINR of HG-CMA, G-CMA and ACMA vs. the number of samples $K$ for 8-PSK and 16-QAM constellations, respectively. We compare the performance of the proposed algorithms HG-CMA and G-CMA with ACMA for different antenna configurations and SNR=30 dB. In both figures we observe that, the larger the number of samples, the better the performance. In Fig. \ref{Fig_5}, in the case of 8-PSK constellation, we observe that HG-CMA provides the best performance. For small number of samples, G-CMA outperforms ACMA. However, for large number of samples ACMA performs better. In Fig. \ref{Fig_6} for 16-QAM, HG-CMA and G-CMA outperform ACMA and the gap is larger for small number of samples and decreases as the number of samples increases.

In Figs. \ref{Fig_7} and \ref{Fig_8} we plot the symbol error rate (SER) of HG-CMA, G-CMA and ACMA vs. SNR for different number of samples $K$ for 8-PSK and 16-QAM constellations, respectively. We considered $M=5$ and $N=7$. In Fig. \ref{Fig_7}, for 8-PSK case, we notice that the proposed HG-CMA provides the best performance. We also observe that G-CMA outperforms ACMA for small number of samples, here $K=20$. However, for large number of samples ACMA performs better than G-CMA for all SNRs. Note that for very large SNR and $K \geq M^2$ it is expected that ACMA outperforms HG-CMA since ACMA in this case provides the optimal (exact in the noiseless case) solution. In the case of 16-QAM in Fig. \ref{Fig_8}, we observe that the proposed HG-CMA and G-CMA algorithms always outperform ACMA, even for large number of samples. Therefore, we can conclude that the proposed HG-CMA and G-CMA are preferable to ACMA in the case of non-constant modulus constellations, i.e. 16-QAM, for any number of samples. In the case of constant modulus constellations, e.g. PSK, HG-CMA and G-CMA are better than ACMA for small number of samples. However, for large number of samples and the range of interest of SNR from $0-30$ dB, HG-CMA and ACMA have close performance and ACMA is better than G-CMA.

%------------------------------------------------------------------------------------------------------------------------
To assess the performance of the adaptive HG-CMA, we consider here, unless stated otherwise, a $5\times5$ MIMO system (i.e. $M=5$), an i.i.d. 8-PSK modulated sequences as input sources, and the processing window size is set equal to $K=2M$. In Fig. \ref{Fig_Time}, we compare the convergence rates and separation quality of adaptive HG-CMA (with different number of rotations per time instant), LS-CMA and adaptive ACMA. One can observe that adaptive HG-CMA outperforms the two other algorithms in this simulation context. Even with only two rotations per time instant, our algorithm leads to high separation quality with fast convergence rate (typically, few tens of iterations are sufficient to reach the steady state level).

In Fig. \ref{Fig_SNR}, the plots represent the steady state SINR (obtained after 1000 iterations) versus the SNR. One can see that the adaptive HG-CMA has no floor effect (as for the LS-CMA and adaptive ACMA) and its SINR increases almost linearly with the SNR in dB.

In Fig. \ref{Fig_M}, the SNR is set equal to $20dB$ and the plots represent again the steady state SINR versus the number of sources $M$. Severe performance degradation is observed (when the number of sources increases) for the LS-CMA and adaptive ACMA while the adaptive HG-CMA performance seems to be unaffected. In Fig. \ref{Fig_K}, the plots illustrate the algorithms performance versus the chosen processing window size\footnote{This concerns only LS-CMA and adaptive HG-CMA as the adaptive ACMA in \cite{Veen_Chap_05} uses an exponential window with parameter $\beta = 0.995$.} $K$. Surprisingly, HG-CMA algorithm reaches its optimal performance with relatively short window sizes ($K$ can be chosen of the same order as $M$).

In the last experiment (Fig. \ref{Fig_QAM}), we consider 16-QAM sources (with non CM property). In that case, all algorithms performance are degraded but adaptive HG-CMA still outperforms the two other algorithms. To improve the performance in the case of non constant modulus signals, one needs to increase the processing window size as illustrated by this simulation result but more importantly, one needs to use more elaborated cost functions which combines the CM criterion with alphabet matching criteria e.g. \cite{Amine_04, Labed_13}.

\begin{figure}%[htbp]
    \centering
    \includegraphics[width=15cm,height=12cm]{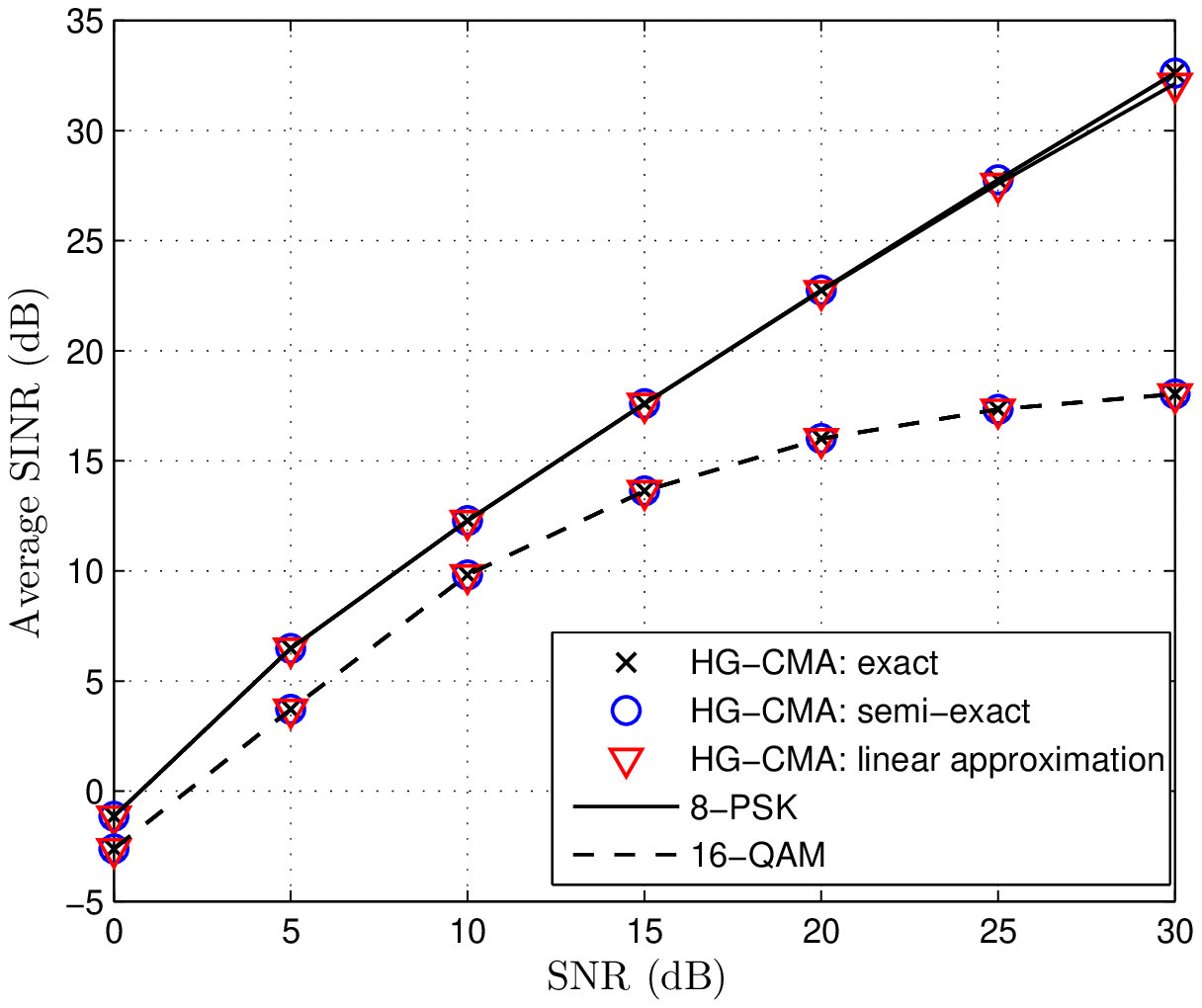}
    \caption{Average SINR of HG-CMA vs. SNR. $M=5$, $N=7$, $K=100$, 8-PSK, 16-QAM, and the number of sweeps is 10.}\label{Fig_1}
\end{figure}

\begin{figure}%[htbp]
    \centering
    \includegraphics[width=15cm,height=12cm]{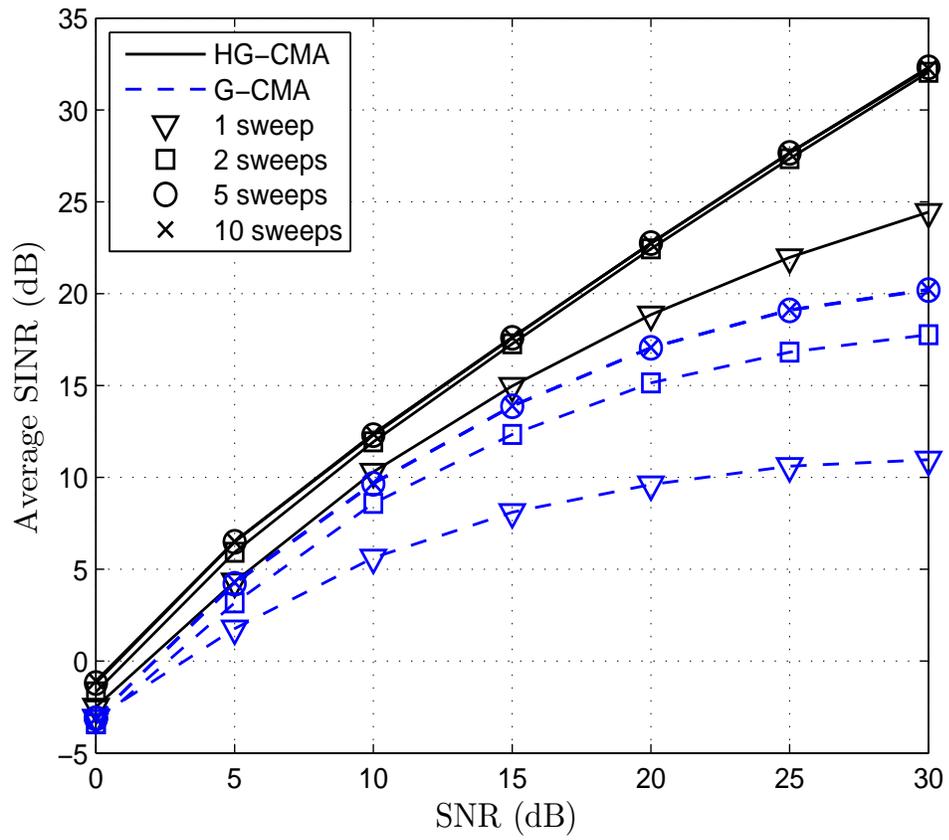}
    \caption{Average SINR of HG-CMA and G-CMA vs. SNR. The effect of the number of sweeps on the performance of G-CMA. $M=5$, $N=7$, $K=100$, and 8-PSK.}\label{Fig_2}
\end{figure}
\begin{figure}%[htbp]
    \centering
    \includegraphics[width=15cm,height=12cm]{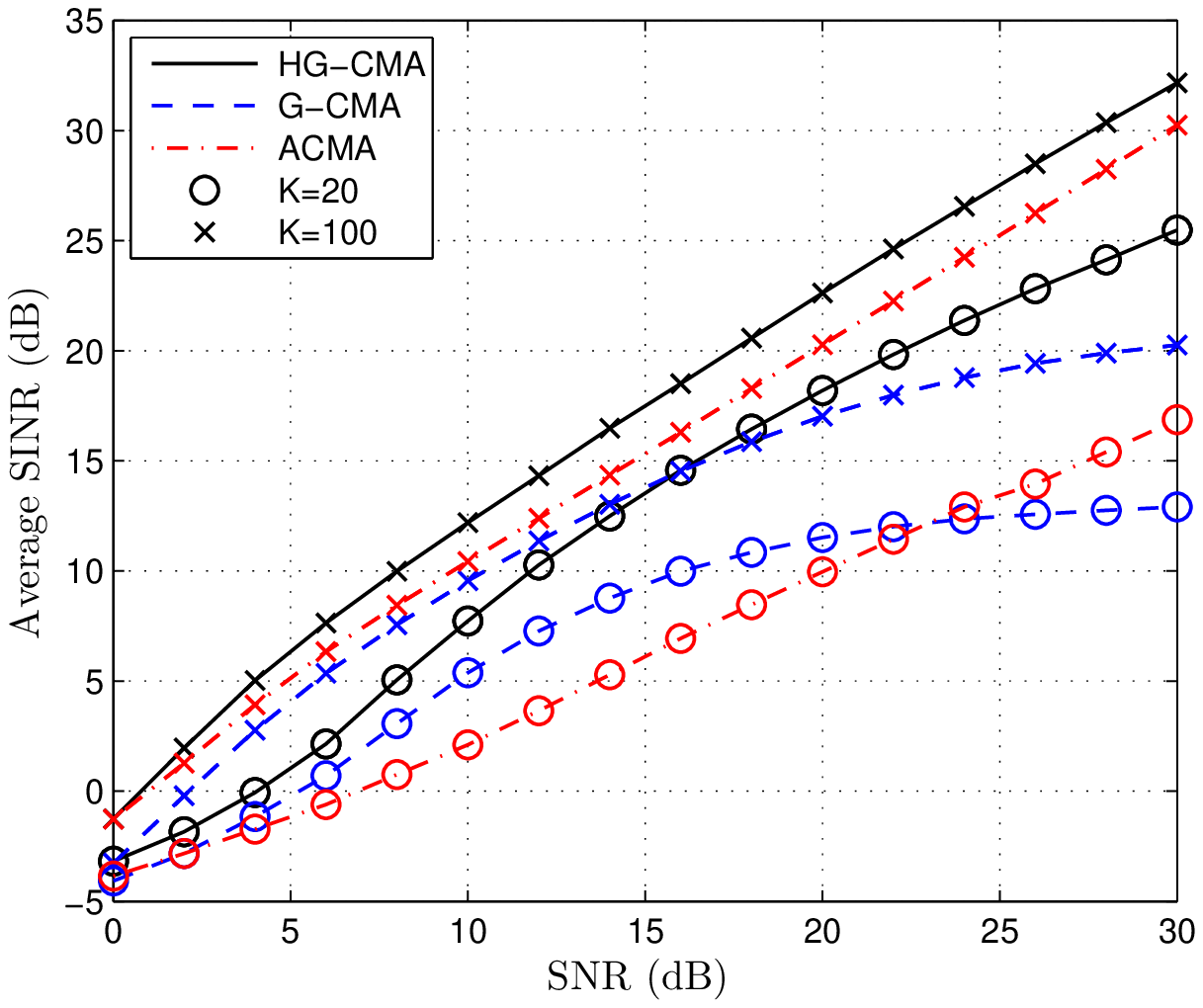}
    \caption{Average SINR of HG-CMA, G-CMA, and ACMA vs. SNR for different numbers of samples $K$. 8-PSK case, $M=5$, $N=7$, and 10 sweeps.}\label{Fig_3}
\end{figure}
\begin{figure}%[htbp]
    \centering
    \includegraphics[width=15cm,height=12cm]{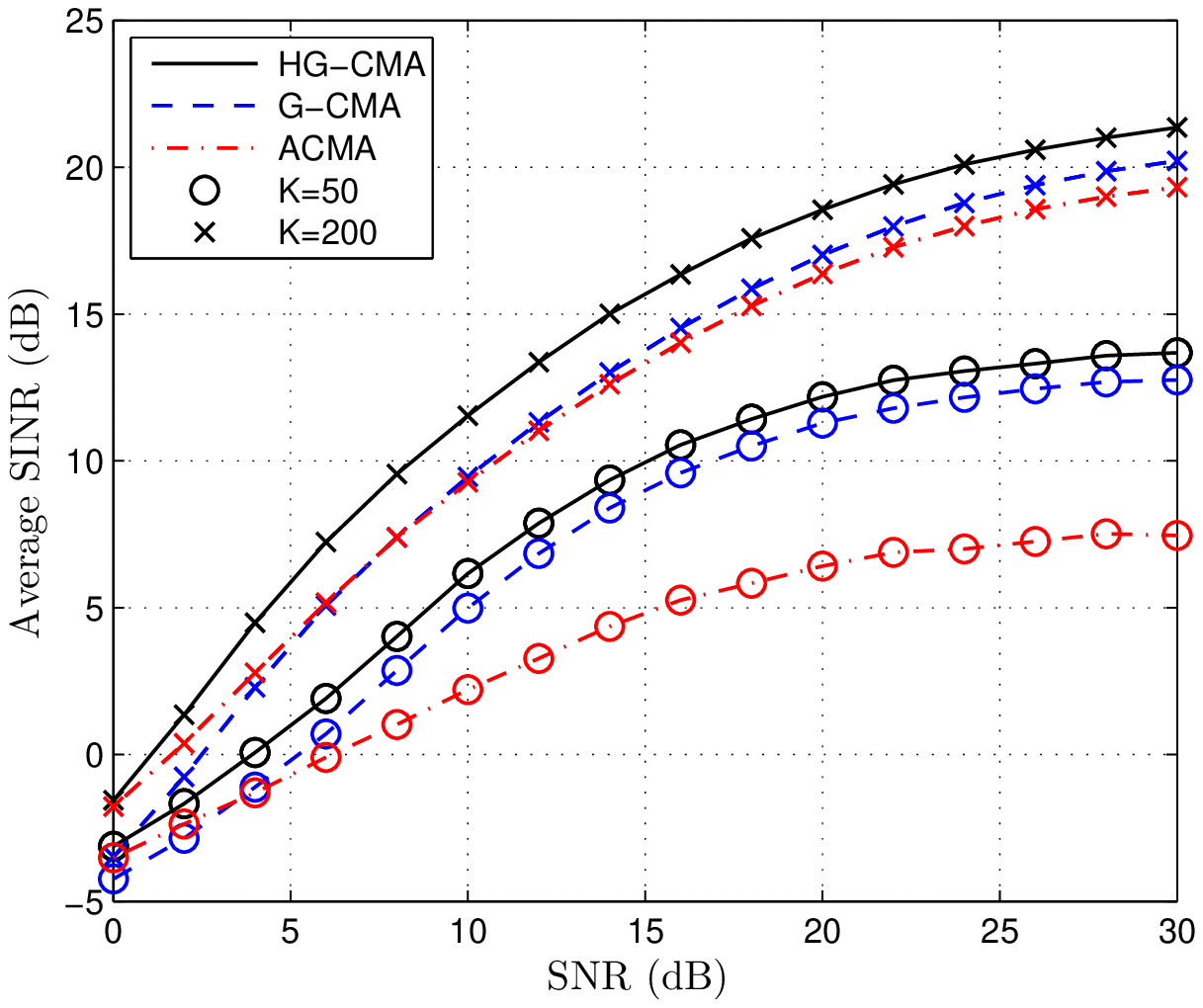}
    \caption{Average SINR of HG-CMA, G-CMA and ACMA vs. SNR for different numbers of samples $K$. 16-QAM case, $M=5$, $N=7$, and 10 sweeps.}\label{Fig_4}
\end{figure}
\begin{figure}%[htbp]
    \centering
    \includegraphics[width=15cm,height=12cm]{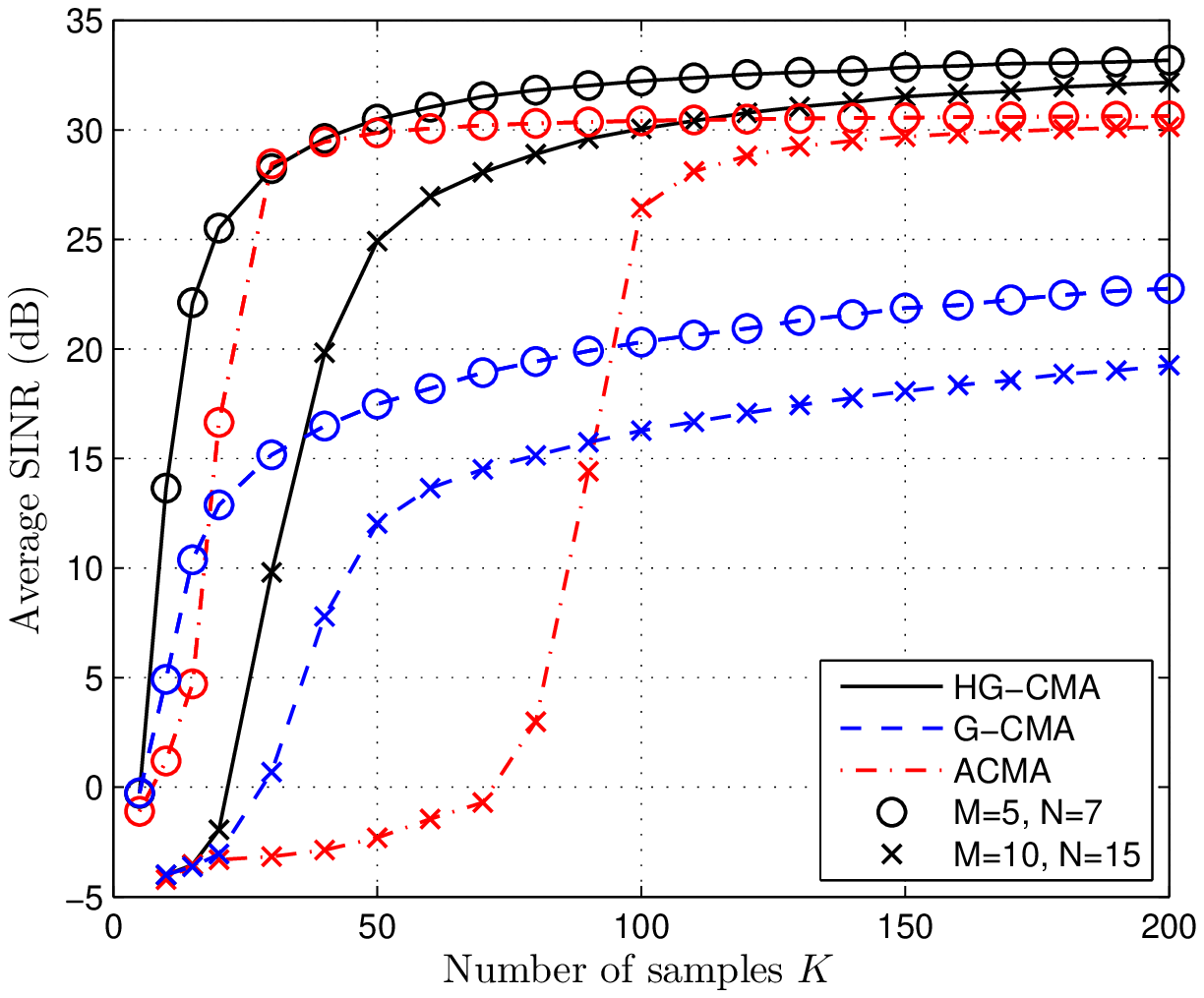}
    \caption{Average SINR of HG-CMA, G-CMA and ACMA vs. the number of samples $K$ for different antenna configurations. 8-PSK case, SNR=30 dB, and 10 sweeps.}\label{Fig_5}
\end{figure}

\begin{figure}%[htbp]
    \centering
    \includegraphics[width=15cm,height=12cm]{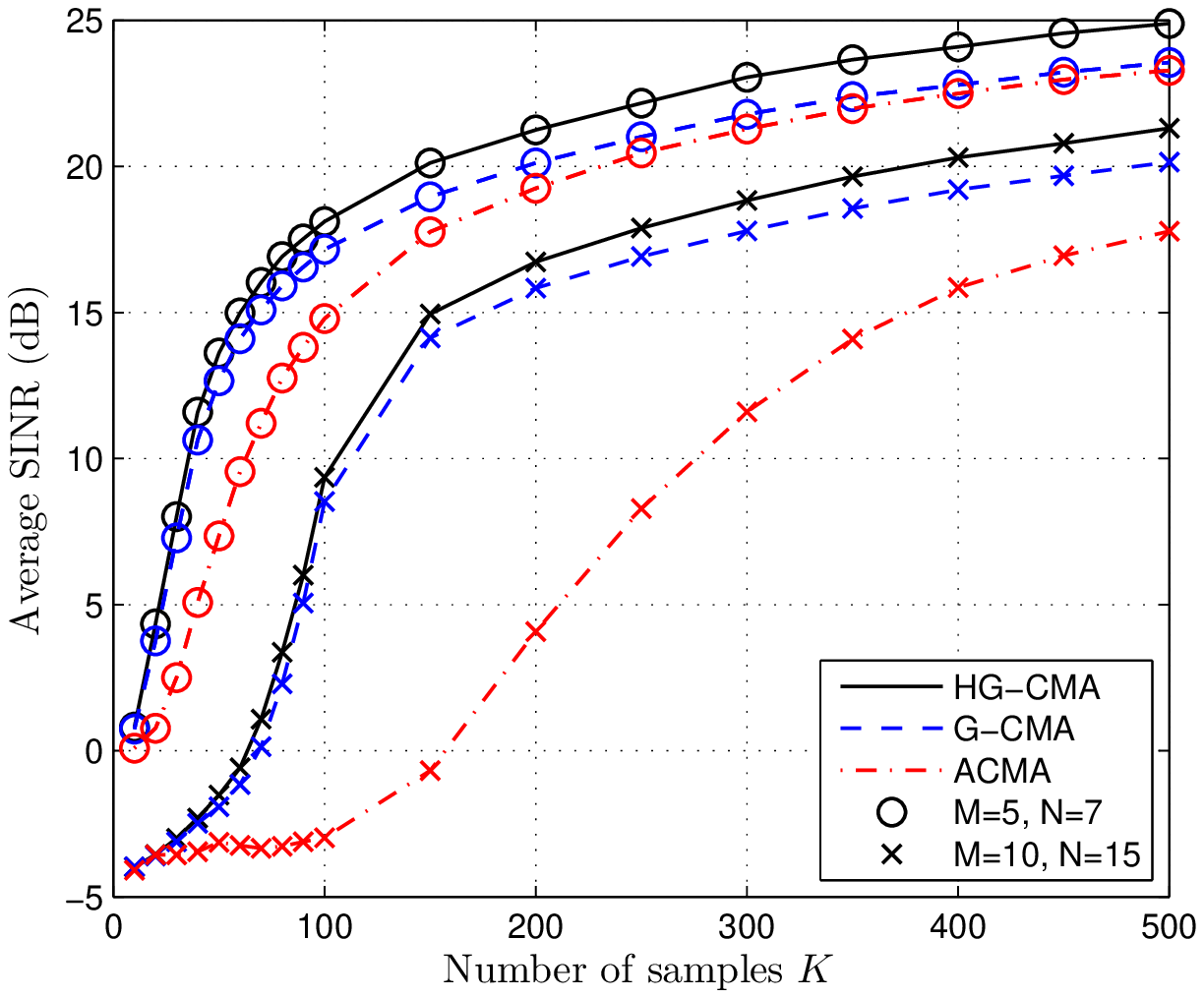}
    \caption{Average SINR of HG-CMA, G-CMA and ACMA vs. the number of samples $K$ for different antenna configurations. 16-QAM case, SNR=30 dB, and 10 sweeps.}\label{Fig_6}
\end{figure}

\begin{figure}%[htbp]
    \centering
    \includegraphics[width=15cm,height=12cm]{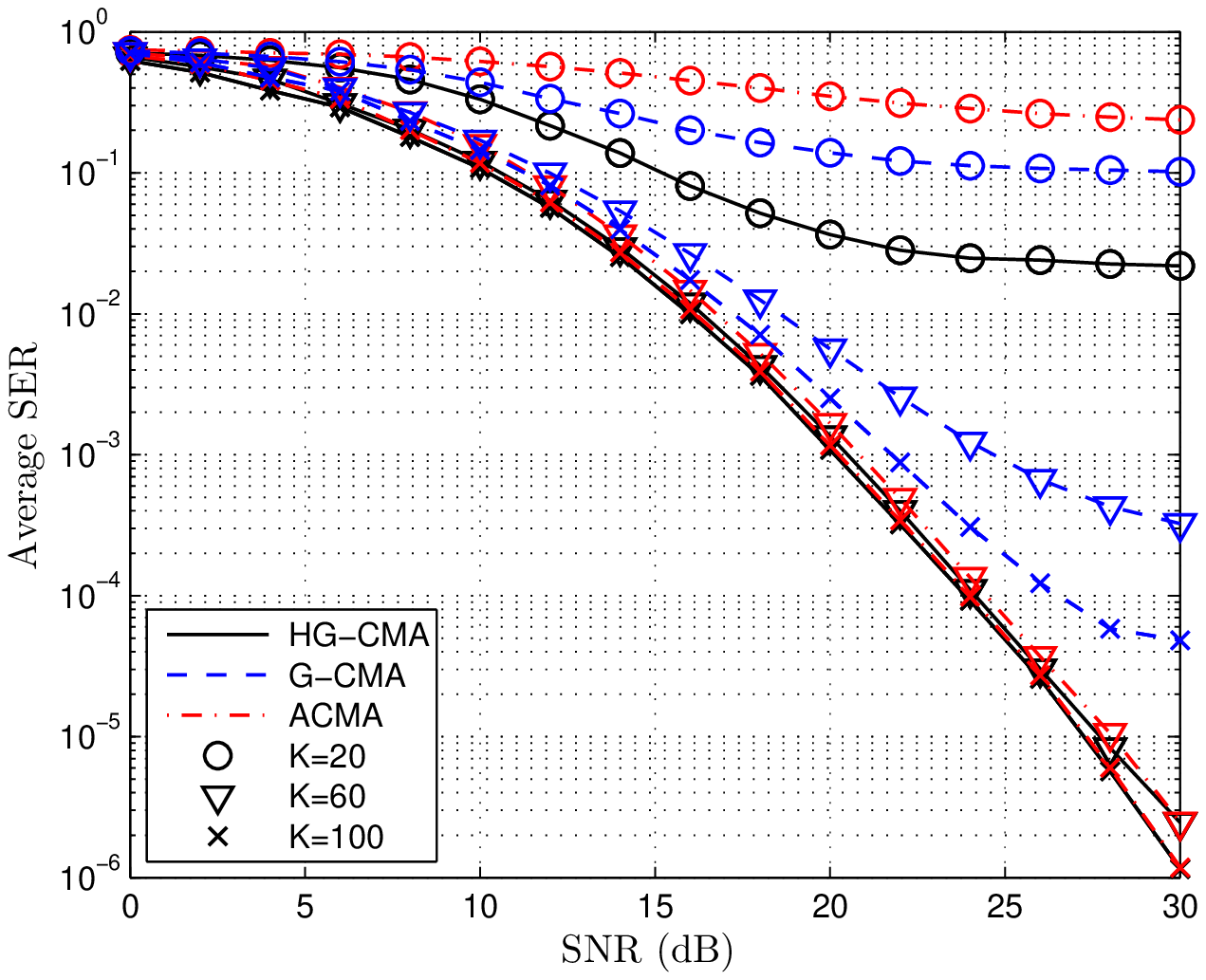}
    \caption{Average symbol error rate of HG-CMA, G-CMA and ACMA vs. SNR for different numbers of samples $K$. 8-PSK case, $M=5$, $N=7$, and 10 sweeps.}\label{Fig_7}
\end{figure}

\begin{figure}%[htbp]
    \centering
    \includegraphics[width=15cm,height=12cm]{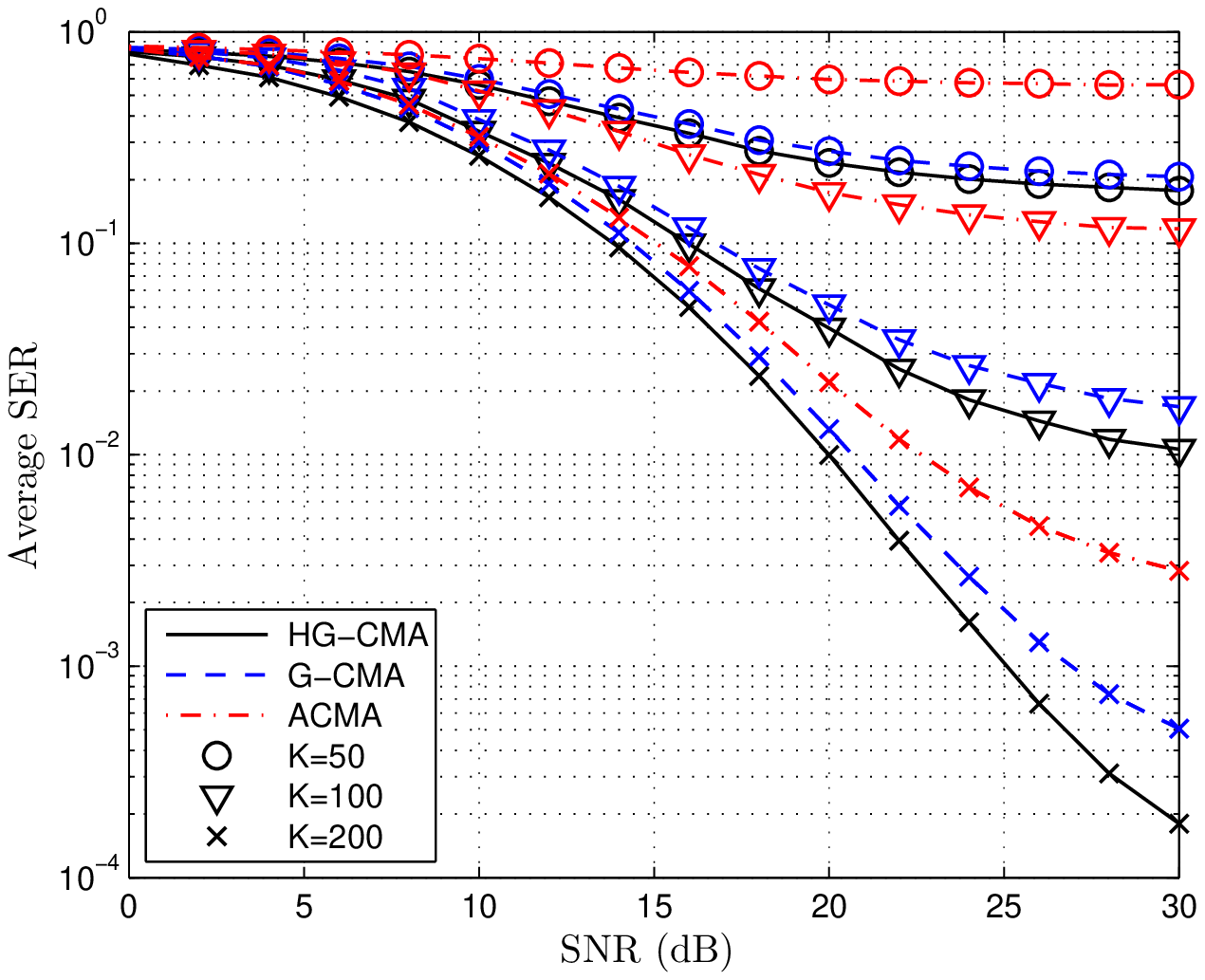}
    \caption{Average symbol error rate of HG-CMA, G-CMA and ACMA vs. SNR for different numbers of samples $K$. 16-QAM case, $M=5$, $N=7$, and 10 sweeps.}\label{Fig_8}
\end{figure}

\begin{figure}%[htbp]
  \centering
	\includegraphics[width=15cm,height=12cm]{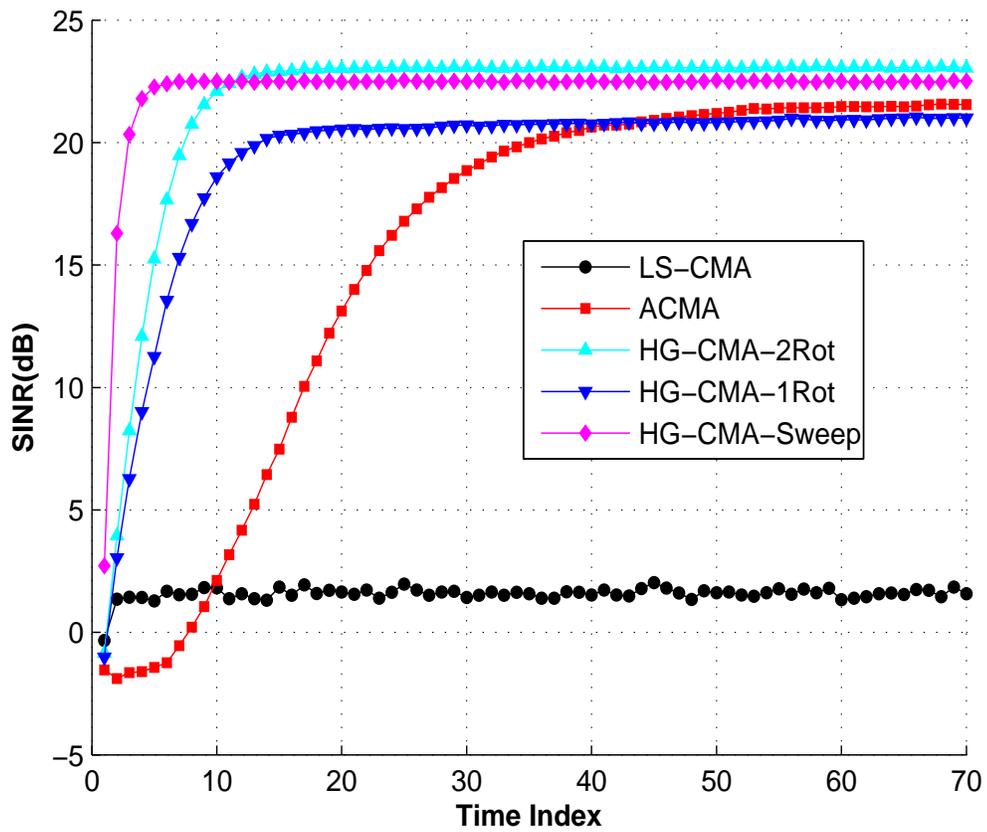}
\caption{SINR vs. Time Index: $SNR=20dB$, $M=N=5$, $K=10$, 8-PSK.}
\label{Fig_Time}
\end{figure}

\begin{figure}%[htbp]
  \centering
	\includegraphics[width=15cm,height=12cm]{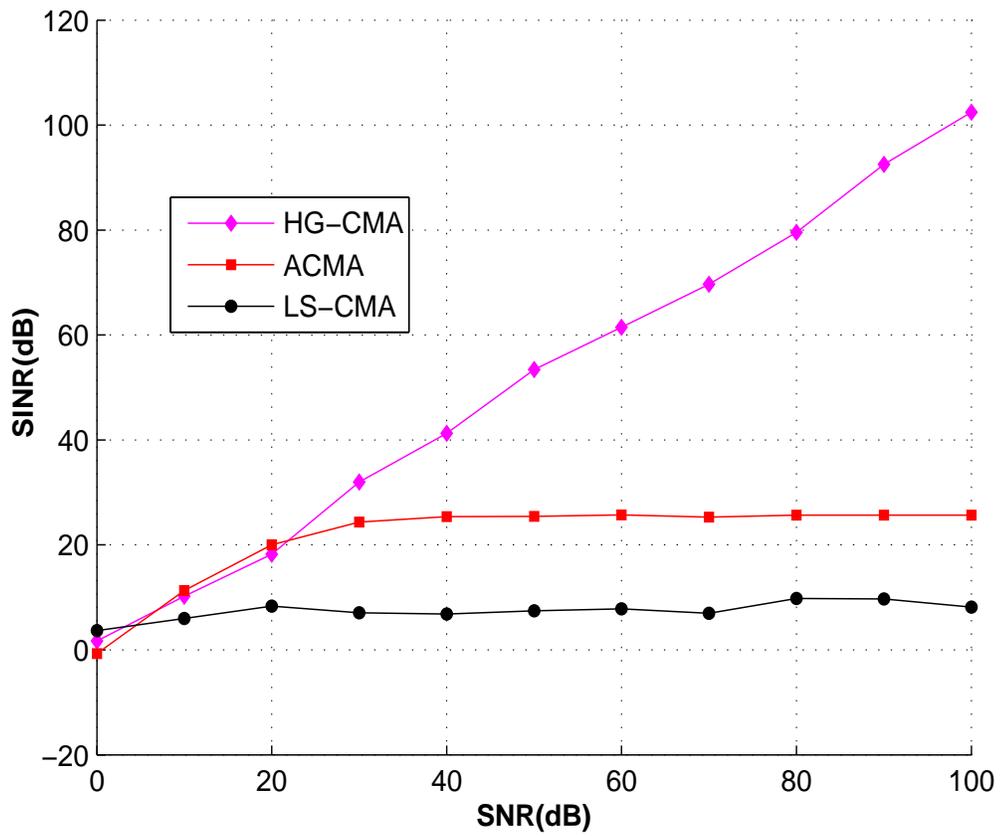}
\caption{SINR vs. SNR: $M=N=5$, $K=10$, 8-PSK.}
\label{Fig_SNR}
\end{figure}

\begin{figure}%[htbp]
  \centering
	\includegraphics[width=15cm,height=12cm]{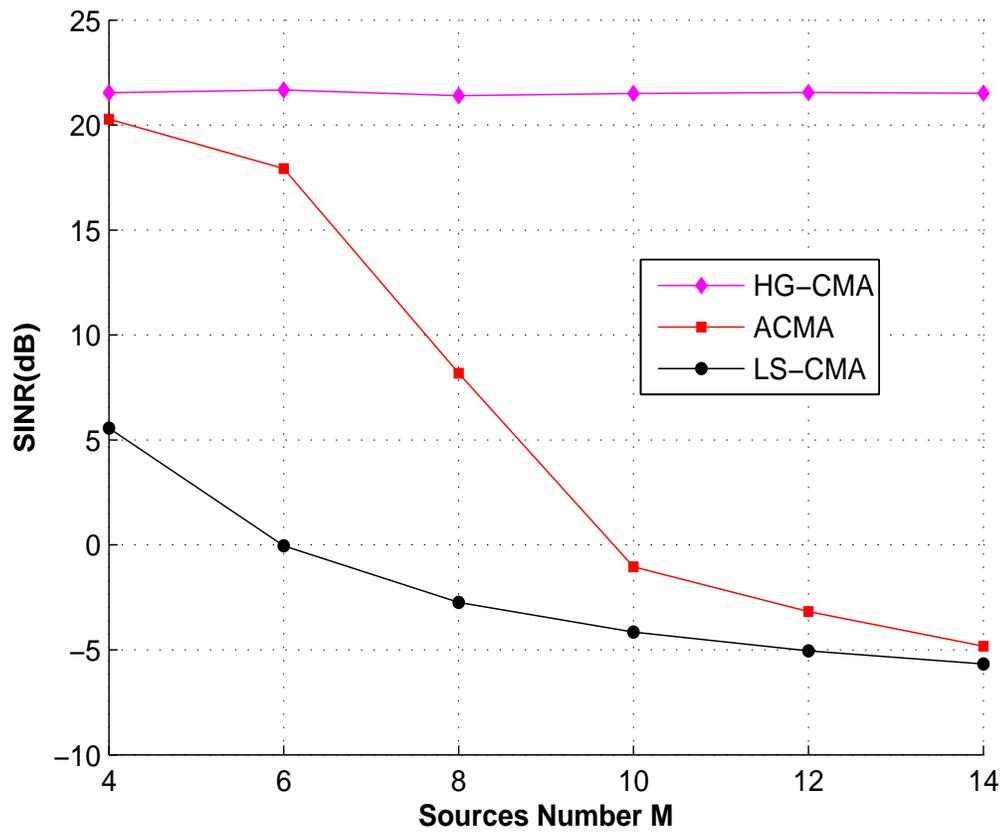}
\caption{SINR vs. Source Number: $SNR=20dB$, $K=2M$, 8-PSK.}\label{Fig_M}
\end{figure}

\begin{figure}%[htbp]
  \centering
	\includegraphics[width=15cm,height=12cm]{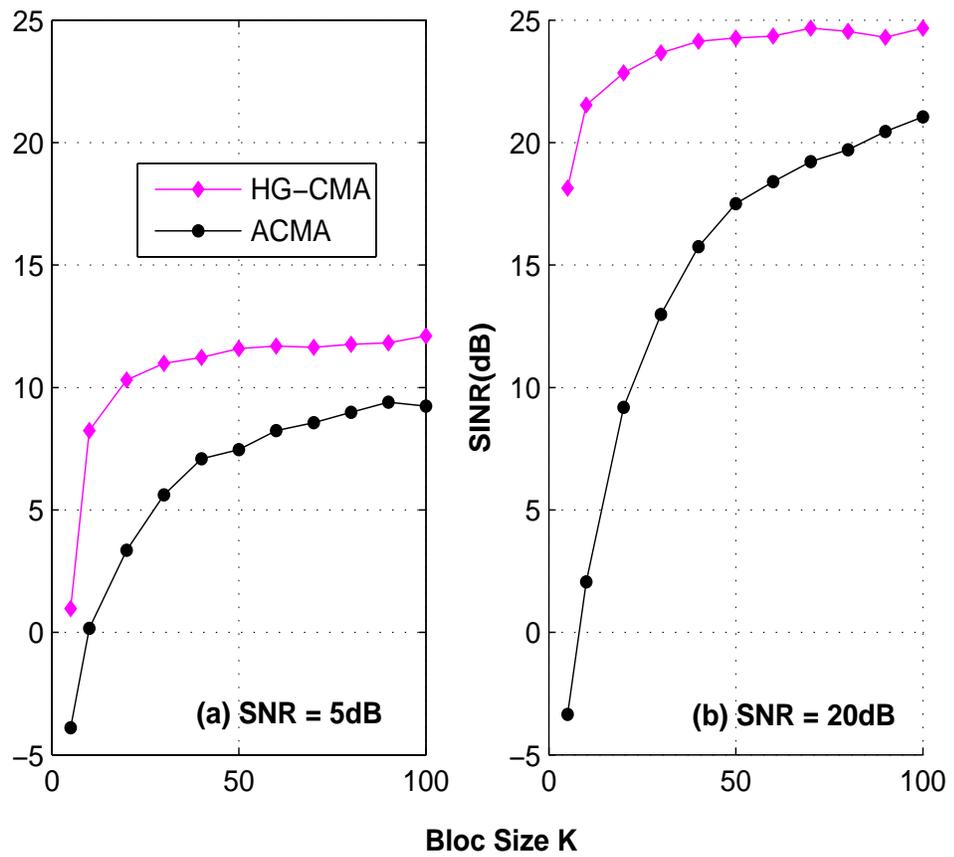}
	\caption{SINR vs. Bloc Size K: $M=N=5$, 8-PSK.}\label{Fig_K}
\end{figure}

\begin{figure}%[htbp]
  \centering
	\includegraphics[width=15cm,height=12cm]{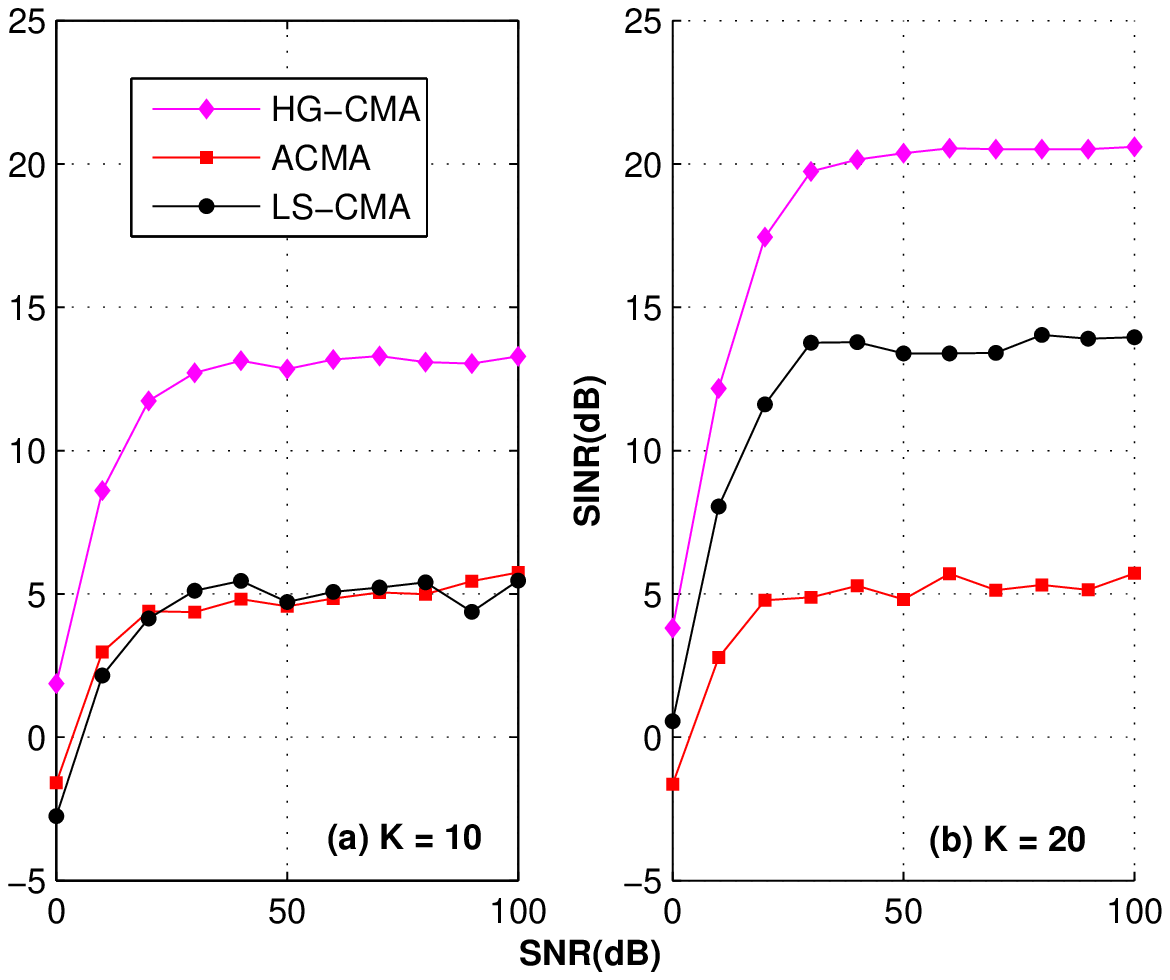}
	\caption{SINR vs. SNR: $M=N=5$, 16-QAM.}\label{Fig_QAM}
\end{figure}

%========================================================================================================================
\section{Conclusion}
\label{Sec:Conclusion}

We proposed two algorithms, G-CMA and HG-CMA, for BSS in the context of MIMO communication systems based on the CM criterion. In G-CMA we combined prewhitening and Givens rotations and in HG-CMA we combined Shear rotations and Givens rotations. G-CMA is appropriate for large number of samples since in this case prewhitening is accurate. However, in the case of small number of samples HG-CMA is preferred since Shear rotations allow to compensate for the prewhitening stage, i.e., reduce the departure from normality. For PSK constellations and small number of samples, we showed that the proposed HG-CMA and G-CMA algorithms are better than the conventional ACMA. However for large number of samples HG-CMA and ACMA have close performance and ACMA outperforms G-CMA. In the case of 16-QAM constellation, HG-CMA and G-CMA outperform largely the conventional ACMA for small number of samples.

Also, for the HG-CMA, a moderate complexity adaptive implementation is considered with the advantages of fast convergence rate and high separation quality. The simulation results  illustrate its effectiveness as compared to the adaptive implementations of ACMA and LS-CMA. They show that the sliding window size can be chosen as small as twice the number of sources without significant performance loss.  Also, they  illustrate the trade off between the convergence rate and the algorithm's numerical cost as a function of the number of used rotations per iteration. As a perspective, the proposed technique can be adapted for the optimization  of more elaborated cost functions which combine the CM criteria with alphabet matching criteria.
%========================================================================================================================
\section{Appendix A} \label{Sec:AppA}
It has been shown in subsection \ref{Sub:ExactSol} that the optimal solution in the sense of minimizing the CM criterion in (\ref{Eq27}) is given by (see equation (\ref{Eq38})):
\begin{equation}\label{A1}
\mathbf{u}=(\mathbf{R}+\lambda \mathbf{J}_{3})^{-1} \mathbf{r}
\end{equation}
where $\lambda$ is the solution of:
\begin{equation}\label{A2}
 \mathbf{u}^T \mathbf{J}_{3} \mathbf{u} = 1 \Longleftrightarrow \mathbf{r}^T (\mathbf{R}+\lambda \mathbf{J}_{3})^{-1}\mathbf{J}_{3}(\mathbf{R}+\lambda \mathbf{J}_{3})^{-1} \mathbf{r}=1
\end{equation}
In the following, we will show that (\ref{A2}) is a $6$-th order polynomial equation.
Let the $3 \times 3$ matrices $\mathbf{U}$ and $\mathbf{\Lambda}=\mbox{diag}\left[\lambda_1~\lambda_2~\lambda_3 \right]$ be the generalized eigenvectors and eigenvalues matrices of the matrix pair ($\mathbf{R}$,~$\mathbf{J}_3$), i.e.
\begin{equation}\label{A3}
    \mathbf{R} = \mathbf{J}_3~\mathbf{U}~\mathbf{\Lambda}~\mathbf{U}^{-1}
\end{equation}
and hence:
\begin{equation}\label{A4}
\left(\mathbf{R}+\lambda\mathbf{J}_3\right)^{-1}=\mathbf{U}\left(\mathbf{\Lambda}+ \lambda\mathbf{I}_3\right)^{-1}\mathbf{U}^{-1}\mathbf{J}_3
\end{equation}
replacing (\ref{A4}) in (\ref{A2}) leads to:
\begin{equation}\label{A5}
\mathbf{r}^T\mathbf{U}\left(\mathbf{\Lambda}+\lambda\mathbf{I}_3\right)^{-2}\mathbf{U}^{-1}\mathbf{J}_3 \mathbf{r}=\mathbf{a}^T\left(\mathbf{\Lambda}+\lambda\mathbf{I}_3\right)^{-2}\mathbf{b}=1
\end{equation}
where $\mathbf{a}^T=\mathbf{r}^T\mathbf{U}=\left[a_1~a_2~a_3\right]$ and $\mathbf{b}=\mathbf{U}^{-1}\mathbf{J}_3 \mathbf{r}=\left[b_1~b_2~b_3\right]^T$. Knowing that $\left(\mathbf{\Lambda}+\lambda\mathbf{I}_3\right)^{-2} = \mbox{diag}\left[(\lambda+\lambda_1)^{-2},~ (\lambda+\lambda_2)^{-2},~(\lambda+\lambda_3)^{-2}\right]$, (\ref{A5}) is rewritten as:
\begin{equation}\label{A6}
\sum_{i=1}^{3}\frac{a_i b_i}{\left(\mathbf{\lambda}+\lambda_i \right)^{2}}=1
\end{equation}
which is equivalent to:
\begin{equation}\label{A8}
\prod_{i=1}^{3}\left(\mathbf{\lambda}+\lambda_i \right)^{2} - \sum_{i=1}^{3} a_i b_i \prod_{j=1, j\neq i}^{3}\left(\mathbf{\lambda}+\lambda_j \right)^{2}=0
\end{equation}
Which is a $6$-th order polynomial equation of the form $P_6(\lambda)=c_0\lambda^6+c_1\lambda^5+c_2\lambda^4+c_3\lambda^3+c_4\lambda^2+c_5\lambda+c_6=0$ with:
\begin{eqnarray} \label{A9}
\begin{array}{lll}
c_0 = 1,~~c_1 = 2\sum_{i=1}^{3}\lambda_i, ~~c_2 = \sum_{i=1}^{3}\left(\lambda_i^2+4\prod_{j=1,j\neq i}^{3}\lambda_j\right)-\mathbf{a}^T\mathbf{b} \nonumber \\
c_3 = 2\sum_{i=1}^{3}\left(\left(\lambda_i^2-a_ib_i\right) \sum_{j=1,j\neq i}^{3}\lambda_j\right), ~~
c_6 = \prod_{i=1}^{3}\lambda_i^2 -\sum_{i=1}^{3}a_ib_i\prod_{j=1,j\neq i}^{3}\lambda_j^2 \nonumber \\
c_4 = \lambda_1^2\lambda_2^2\left(1+\lambda_3^2\right)+4\prod_{i=1 }^{3}\lambda_i\sum_{i=1}^{3}\lambda_i-  \sum_{i=1}^{3}a_ib_i\left(\sum_{j=1,j\neq i}^{3}\lambda_j^2+4\prod_{j=1,j\neq i}^{3}\lambda_j\right)\nonumber \\
c_5 = 2\left(\prod_{i=1}^{3}\lambda_i\right)\left(\sum_{i=1}^{3}\prod_{j=1,j\neq i}^{3}\lambda_j\right) -\sum_{i=1}^{3}a_ib_i\left(\sum_{j=1,j\neq i}^{3}\lambda_j\right)\left(\prod_{j=1,j\neq i}^{3}\lambda_j\right)\nonumber \\

\end{array}
\end{eqnarray}

Using the same reasoning, we can find the coefficients of the $4$-th order polynomial equation in (\ref{Eq51});  $P_4(\lambda)=c_0\lambda^4+c_1\lambda^3+c_2\lambda^2+c_3\lambda^1+c_4=0$.
\begin{eqnarray} \label{A10}
\begin{array}{lll}
c_0 = 1, ~~c_1 = 2\sum_{i=1}^{2}\tilde{\lambda_i}, ~~ c_2 = \sum_{i=1}^{2}\tilde{\lambda_i}^2+4\prod_{j=1,j\neq i}^{2}\tilde{\lambda_j}-\tilde{\mathbf{a}}^T \tilde{\mathbf{b}} \nonumber \\
c_3 = 2\sum_{i=1}^{2}\left(\tilde{\lambda_i}^2- \tilde{a}_i \tilde{b}_i\right) \sum_{j=1,j\neq i}^{2}\tilde{\lambda_j}, ~~
c_4 = \prod_{i=1}^{2}\tilde{\lambda_i}^2 -\sum_{i=1}^{2} \tilde{a}_i \tilde{b}_i\prod_{j=1,j\neq i}^{2}\tilde{\lambda_j}^2 \nonumber
\end{array}
\end{eqnarray}
with $\tilde{\mathbf{a}}^T=\tilde{\mathbf{r}}^T\tilde{\mathbf{U}}=\left[\tilde{a}_1~\tilde{a}_2\right]$ and $\tilde{\mathbf{b}}=\tilde{\mathbf{U}}^{-1}\mathbf{J}_2 \tilde{\mathbf{r}}=\left[\tilde{b}_1~\tilde{b}_2\right]^T$.
Where the $2 \times 2$ matrices $\tilde{\mathbf{U}}$ and $\tilde{\mathbf{\Lambda}} = \mbox{diag} \left[\tilde{\lambda}_1 ~\tilde{\lambda}_2\right]$ represent the generalized eigendecomposition of the matrix pair ($\tilde{\mathbf{R}}$,~$\mathbf{J}_2$).
%========================================================================================================================

%========================================================================================================================
\end{document}